\documentclass[a4paper,11pt]{amsart}

\usepackage{amsmath,amssymb, amsthm, amsfonts}
\usepackage{enumerate}
\usepackage{xcolor}
\usepackage{tikz}
\usepackage{enumerate}
\usepackage{calc}
\usepackage{marginnote}
\usepackage{caption}
\usepackage{subcaption}

\usepackage{comment}

\usepackage{hyperref}
\hypersetup{ colorlinks = true,linkcolor = {blue},citecolor = {red} }

\theoremstyle{definition}
\newtheorem{remark}{Remark}
\newtheorem{proposition}{Proposition}[section]
\newtheorem{lemma}[proposition]{Lemma}
\newtheorem{theorem}{Theorem}
\newtheorem*{theorem*}{Theorem}

\newcommand{\ZZ}{\mathbb{Z}}

\newcommand{\RR}{\mathbb{R}}
\newcommand{\CC}{\mathbb{C}}

\DeclareMathOperator{\ad}{ad}

\DeclareMathOperator{\Diff}{Diff}

\DeclareMathOperator{\PSL}{PSL}
\DeclareMathOperator{\vect}{vect}
\DeclareMathOperator{\Stab}{Stab}
\DeclareMathOperator{\erf}{erf}

\newcommand{\OO}{\mathcal O}
\newcommand{\g}{\mathfrak g}

\newcommand{\vir}{\mathfrak{vir}}

\newcommand{\del}{\partial}

\newcommand{\cy}{\color{cyan}}
\newcommand{\red}{\color{red}}

\begin{document}
\title[]{Towards bosonization of Virasoro coadjoint orbits }

\author[A.\ Alekseev]{Anton Alekseev$^1$}
\address{$^1$Section of Mathematics, University of Geneva, Rue du Conseil G\'en\'eral 7-9, 1211 Geneva, Switzerland}
\email{anton.alekseev@unige.ch}

\author[O.\ Chekeres]{Olga Chekeres$^2$}
\address{$^2$Department of Mathematics, University of Connecticut, 341 Mansfield Rd, Storrs, 06269 Connecticut, USA}
\email{olga.chekeres@uconn.edu}

\author[D.\ R.\ Youmans]{Donald R.\ Youmans$^{3,4}$}
\address{$^3$Albert Einstein Center for Fundamental Physics,
Institute for Theoretical Physics,
University of Bern, Sidlerstrasse 5, 3012 Bern, Switzerland}\textbf{}

\address{$^4$ Heidelberg STRUCTURES Excellence Cluster, Heidelberg University, Berliner Str.\ 47, 69120 Heidelberg, Germany}\textbf{}
\email{youmans@uni-heidelberg.de}


\begin{abstract}
    We show that Schwarzian theories associated to certain hyperbolic and parabolic Virasoro coadjoint orbits admit  bosonization, {\em i.e.} a global $S^1$-equivariant Darboux chart in which the corresponding path integral becomes Gaussian.
    In this chart, correlation functions of bilocals, time-ordered and out-of-time ordered, can be computed explicitly. 
    We conjecture that a similar global chart exists for the Teichm\"uller orbit. 
\end{abstract}

\dedicatory{In memory of Krzysztof Gawedzki, teacher and friend}

\maketitle

\tableofcontents

\section{Introduction}

Coadjoint orbits of the Virasoro group have experienced a surge of interest in recent years due to their independent appearance in two-dimensional gravity models \cite{Maldacena_Stanford_Yang,Mertens_Turiaci} and in condensed matter   SYK-models \cite{Maldacena_Stanfod}.

Associated to any Virasoro coadjoint orbit, there is a physical system whose configuration space of fields is given by the orbit itself.
The action functional of the model is given by the moment map of  rigid $S^1$-rotations on the orbit:
\begin{equation}
    \mu_{b_0}(f) = \int_{S^1} \left(b_0(f(x)) f'^2(x) + \frac{c}{12} S(f)\right) dx,
    \label{eq:action}
\end{equation}
where $f(x+2\pi)=f(x)+2\pi$ is a diffeomorphism of the circle, and $S(f)$ is the Schwarzian derivative:
\begin{equation*}
    S(f) = \frac{f'''}{f'} - \frac{3}{2} \left(\frac{f''}{f'} \right)^2.
\end{equation*}
In the action \eqref{eq:action}, the periodic function $b_0(x)$ characterizes the orbit. In this paper, we will consider Virasoro orbits with $b_0={\rm const}$. Even in this case, the action \eqref{eq:action} defines a family of rather complicated nonlinear theories.

Nevertheless, it was realized in \cite{Stanford_Witten} that  these theories are one-loop exact. 
In this article we go one step further and show that for a wide class of Virasoro coadjoint orbits the action \eqref{eq:action} actually defines a free theory. 
This is an instance of bosonization which unravels the hidden structure of a non-interacting theory.
Classical examples of bosonization were provided by Gawedzki-Kupiainen in the case of WZW models \cite{GK} (see also \cite{Alekseev_Shatashvili_pathintegral_quantization, GMOMS}).

Our work is inspired by some classical results in equivariant symplectic geometry. In more detail, consider a symplectic manifold $M$ with a Hamiltonian $S^1$-action and an isolated fixed point $p \in M$. Then, the equivariant Darboux theorem states that there is a neighborhood of $p$ and a coordinate system such that the symplectic form is constant, the $S^1$-action linear, and the moment map of the $S^1$-action quadratic. Furthermore, if $p$ is an extremum 
of the moment map, the size of the Darboux chart is controlled by the next fixed point of the $S^1$-action \cite{Karshon-Tolman}. If there are no other fixed points, this equivariant Darboux chart is global.

On a Virasoro coadjoint orbit with $b_0={\rm const}$ there is a unique fixed point of rigid rotations $r(x)=x+\theta$ which is given by $f(x)=x$ \cite{Dai-Pickrell}. 
For a large class of orbits which include certain hyperbolic and elliptic orbits as well as the Teichm\"uller orbit, this fixed point is a maximum of the moment map. 

In this paper, we show that hyperbolic Virasoro orbits with constant representative $b_0<0$ and the parabolic orbit with $b_0=0$ do admit a global Darboux chart. Hence, the corresponding Schwarzian theories admit bosonization. Furthermore, observables of the model \eqref{eq:action} are bilocals $\OO(x,y)$. In  Darboux coordinates, they take an exponential form, and their correlation functions are given by Gaussian integrals. In particular, correlation functions of two bilocals $\OO(x_1,x_2)\OO(x_3,x_4)$ can be computed for time-ordered $x_1 < x_2 < x_3 < x_4$ or out-of-time-ordered $x_1 < x_3 < x_2 < x_4$ configurations.
Comparison of these correlation functions is the standard probe for quantum chaos, see {\em e.g.}\ \cite{Sonner}. 
Such a comparison, however, is beyond the scope of this article.
Note that in the finite dimensional case correlation functions can be explicitly computed on coadjoint orbits of the unitary group $U(n)$ using Gelfand-Zeitlin Darboux coordinates (see \cite{S}).

We put forward the following conjecture:

\vskip 0.2cm

{\bf Conjecture:} Elliptic Virasoro orbits with $0<b_0<\frac{c}{24}$ and  the exceptional Teichm\"uller orbit ($b_0 = \frac{c}{24}$) admit global equivariant Darboux charts. 

\vskip 0.2cm

In the case of the Teichm\"uller orbit, we give its presentation as a codimension 2 submanifold in an (infinite dimensional) Darboux chart. We believe  that this presentation may give a good starting point for proving the conjecture.

The article is structured as follows: In Section \ref{sec:geometry_of_virasoro_coadjoint_orbits}, we recall some basic facts about the geometry of Virasoro coadjoint orbits. 
In Section \ref{sec:darboux_coordinates_for_hyperbolic_orbits}, we prove the main theorem of this article: 
\vskip 0.2cm

    {\bf Theorem:} Hyperbolic Virasoro coadjoint orbits $\OO_{b_0}$ with  $b_0 < 0$ and the parabolic coadjoint orbit $\OO_0$ admit  global equivariant Darboux charts.

\vskip 0.2cm
The proof is given by an explicit construction. 
In Section \ref{sec:problems_of_the_teichmueller_orbit}, we discuss the geometry of the Teichm\"uller orbit and put forward a conjecture on the existence of the global equivariant Darboux chart.
Section \ref{sec:partition_function_and_correlation_functions_of_bilocal_operators} is devoted to the computation of the partition function, and of correlation functions for one and two bilocals  (time-orderd and out-of-time ordered) of the Schwarzian theory associated to a hyperbolic coadjoint orbit. \bigskip

\noindent\textbf{Acknowledgements.}\
The authors would like to thank A.\ Altland, N.\ Dondi, E.\ Meinrenken, N.\ Delporte, S.\ Shatashvili, J.\ Sonner, and A.\ Thomsen for interesting and insightful discussions. 

Research of AA was supported in part by the grants 182767, 208235 and 200400, 
and by the NCCR SwissMAP of the Swiss National Science Foundation (SNSF).
The work of DRY was supported by the NCCR SwissMAP of the SNSF and in part by the Deutsche Forschungsgemeinschaft (DFG, German Research Foundation) under Germany's Excellence Strategy EXC 2181/1 - 390900948 (the Heidelberg STRUCTURES Excellence Cluster).

\section{Geometry of Virasoro coadjoint orbits}\label{sec:geometry_of_virasoro_coadjoint_orbits}

In this Section, we recall the definition and basic properties of coadjoint orbits of the central extension of the group of orientation preserving diffeomorphisms of the circle. Usually, these orbits are referred to as Virasoro coadjoint orbits.

Let $\Diff^{+}(S^1)$ be the  group of orientation preserving diffeomorphisms of the circle and $\widetilde\Diff^{+}(S^1)$ its universal cover.
Elements of $\widetilde\Diff^{+}(S^1)$ are functions $f \colon \RR \to \RR$ satisfying $f'(x) > 0$ and the quasi-periodicity condition 
$$
f(x + 2\pi) = f(x) + 2\pi.
$$

The Lie algebra $\vect(S^1)$ of $\Diff^{+}(S^1)$ is given by vector fields on the circle. 
It admits a universal central extension, the Virasoro algebra $\vir$.
As a vector space $\vir = \vect(S^1) \oplus \RR$.
To simplify notation, we denote the generator of the central $\mathbb{R}$ by $1$.
Elements of $\vir$ are pairs $(L(x)\del_x,k)$ where $L(x)$ is a periodic function and $k \in \RR$.
The Lie bracket is defined by formula
\begin{equation}
  [L_1+k_1,L_2+k_2] = [L_1,L_2] + \frac{1}{12} \int_{S^1} dx\ L_1'''(x)L_2(x).
  \label{eq:adj}
\end{equation}
In particular, 
if one introduces a basis $\ell_m = \frac{1}{2\pi}e^{imx}\del_x$ of (the complexification of) $\vect(S^1)$, then one recovers the standard representation of the Virasoro algebra
\begin{equation}
  2\pi i[\ell_m,\ell_n] = (m-n)\ell_{m+n} + \frac{1}{12} m^3 \delta_{m+n,0},
  \label{eq:virasoro_algebra}
\end{equation}
\begin{remark} 
  In the literature 
  it is common to choose the generator of the central $\RR$ 
  to be $c$, and to normalize the central term in \eqref{eq:virasoro_algebra} by $i c/24\pi$, see {\em e.g.}\ \cite{W1}. If one then considers the basis \mbox{$L_n = i e^{inx}\del_x$}, with $L_0$ shifted by $c/24$, 
  Equation \eqref{eq:virasoro_algebra} takes on the more familiar form
  \begin{equation*}
    [L_m, L_n] = (m-n) L_{m+n} + \frac{c}{12} (m^3-m) \delta_{m+n,0}.
  \end{equation*}
  In this case, $c$ is known as the central charge.
\end{remark}

Elements of the (smooth) dual $\vir^*$ can be identified with pairs $(b,-c)$ where $b = b(x)dx^2$ is a quadratic differential on the circle and $c \in \RR$. 
The pairing between $\vir$ and $\vir^*$ is defined by
\begin{equation}
  \left\langle b-c,L+k \right\rangle = \int_{S^1} b(x)L(x) dx - c k.
\end{equation}
From \eqref{eq:adj} one finds the infinitesimal coadjoint action of $\vect(S^1) \subset \vir$
\begin{equation}
  \ad^*_L(b - c) =  \left(2 L'(x) b(x) + L(x) b'(x) + \frac{c}{12}L'''(x) \right) dx^2,
  \label{eq:inf_coadj}
\end{equation}
which integrates to 
\begin{equation}
  Ad^*_{f^{-1}}(b - c) = b^f - c = \left( f'^2 b(f(x)) + \frac{c}{12} S(f)\right) dx^2 - c,
  \label{eq:coad}
\end{equation}
where $S(f)$ is the Schwarzian derivative of $f$
\begin{equation*}
  S(f) = \frac{f'''}{f'} - \frac32 \left( \frac{f''}{f'} \right)^2. 
\end{equation*}
Coadjoint orbits are hence of the form 
\begin{equation*}
  \OO_b = \{ b^f(x)dx^2 \mid f \in \widetilde\Diff^+(S^1) \} \cong \widetilde\Diff^+(S^1) / \widetilde\Stab(b) \cong \Diff^+(S^1) / \Stab(b).
\end{equation*}
Here $\widetilde\Stab(b) \subset \widetilde\Diff^{+}(S^1)$ is the stabilizer subgroup of $b$ under the action of $\widetilde\Diff^{+}(S^1)$ and $\Stab(b)$ is the stabilizer subgroup of $b$ under the action of $\Diff^+(S^1)$. Since the action \eqref{eq:coad} of $\widetilde\Diff^{+}(S^1)$ on quadratic differentials factors through the action of $\Diff^+(S^1)$, the quotient spaces are equal to each other.

In what follows we will restrict our consideration to the case of $b$ being a constant quadratic differential, {\em i.e.}\ $b(x) = b_0$.
It is known that for generic values of $b_0$, the stabilizer subgroup $\Stab(b)$ is isomorphic to $S^1$ while for the special values 
\begin{equation}
  b_0 = \frac{c k^2}{24}, 
\end{equation}
it is isomorphic to a $k$-fold cover of $\PSL(2,\RR)$.
Orbits with $b_0<0$ are called hyperbolic, with $b_0>0$ and $b_0 \neq ck^2/24$ elliptic, with $b_0=ck^2/24$ exceptional, and the orbit with $b_0=0$ is parabolic
(see \cite{Dai-Pickrell, Kirillov, LP, Segal, W1} for details of the classification).

Recall that coadjoint orbits of Lie groups carry canonical Kirillov-Kostant-Souriau (KKS) symplectic structures \cite{KKS}. In more detail, let $G$ be a Lie group with Lie algebra $\g$, and let $\OO_{\xi} \subset \g^*$ be a coadjoint orbit through a point $\xi \in \g^*$. 
The canonical $G$-invariant symplectic form on $\OO_\xi$ at the point $\xi$ is given by 
\begin{equation}
    \varpi_{\xi}(x,y) = \langle\xi, [x,y] \rangle, \qquad x,y \in \g.
\end{equation}
Here $x,y$ in $\varpi_\xi(x,y)$ stand for fundamental vector fields on $\OO_\xi$.

Since $\OO_{\xi}$ is diffeomorphic to $G / \Stab(\xi)$, the 2-form $\varpi_\xi$  can be pulled back to $G$ under the projection $\pi \colon G \to G / \Stab(\xi) \cong \OO_{\xi}$:
\begin{equation}
    \omega = \pi^*\varpi \in \Omega^2(G).
\end{equation}
At the identity, $e \in G$, one has 
\begin{equation*}
    \omega_e(x,y) = \langle \xi, [x,y] \rangle, \qquad x,y \in T_eG = \g.
\end{equation*}
To find the expression of $\omega$ at any point $g$, one uses left translation on the group: for $x_g, y_g \in T_gG$ two left-invariant vector fields one has
\begin{equation*}
    \omega_g(x_g,y_g) =(L_{g^{-1}})^*\omega_e(x_g,y_g) = \langle \xi, [(L_{g^{-1}})_*x_g, (L_{g^{-1}})_* y_g] \rangle.
\end{equation*}
This allows us to write $\omega_g$ very concisely in terms of the Maurer-Cartan form $Y
\in \Omega^1(G,\g)$:
\begin{equation}
    \omega_g = \frac12\langle \xi, [Y(g),Y(g)] \rangle.
    \label{eq:KKS_general}
\end{equation}
\begin{remark}
 Recall that the Maurer-Cartan form $Y$ satisfies the following identities: $\langle Y(g), x_g\rangle=x$ and
  \begin{subequations}   \label{eq:properties_MC}
  \begin{align}
  &\delta Y(g) + \frac12 [Y(g),Y(g)] = 0,\\
  &\delta \xi^g = ad^*(Y(g))\xi^g.
\end{align}
\end{subequations}
Here and henceforth $\delta$ stands for the de Rham differential on the group and $\xi^g = Ad^*_g \xi$.
This allows one to write $\omega_g$ as:
\begin{equation}
    \omega_g= \frac12 \langle \delta \xi^g, Y(g) \rangle
\end{equation}
\end{remark}

Returning to diffeomorphisms of the circle, consider the Maurer-Cartan form on the Virasoro group, {\em c.f.}\ \cite{Aratyn_Nissimov_Pacheva_1990,Barnich_Gonzalez_Saldago_Rebolledo})
%
\begin{equation}
  \hat Y = \frac{\delta f}{f'} + \delta k + \frac{1}{24}\int_{0}^{2\pi}dx\ \frac{\delta f}{f'} \left( \frac{f''}{f'} \right)'.
\end{equation} 
Here $k$ denotes the coordinate on the central $\mathbb{R}$, and $\delta k$ is the corresponding one-form.

For $\hat Y$, equations \eqref{eq:properties_MC} take the form
\begin{subequations}
  \begin{align}
  &\delta \hat Y + \frac12 [\hat Y,\hat Y] = 0,\\
  &\delta (b^f - c) = ad^*(\hat Y) b^f,
\end{align}
\label{eq:MC_Y}
\end{subequations}
where $[\cdot, \cdot]$ is the Lie bracket of $\vir$, {\em c.f.}\ \eqref{eq:adj}, and the coadjoint action is given by \eqref{eq:inf_coadj}.
The KKS symplectic form at a point $b - c \in \vir^*$ is then given by
\begin{equation}
  \begin{split}
  \omega_{b^f - c} &= \frac12 \left\langle b^f - c, [\hat Y, \hat Y] \right\rangle  \\
  &= \frac12 \int_{S^1} \left( \delta b^f \wedge \frac{\delta f}{f'} \right) dx \\
  &= \int_{S^1}\left(  b_0 \delta f' \wedge \delta f + \frac{c}{24} \delta \log f' \wedge \delta \log(f')'  \right)dx.
  \end{split}
  \label{eq:KKS_symplectic_form}
\end{equation}
A direct proof of the identities \eqref{eq:MC_Y} and \eqref{eq:KKS_symplectic_form} is given in Appendix \ref{app:proofs_schwarzian}.

Coadjoint orbits $\OO_{b_0} \cong \Diff^{+}(S^1)/\Stab(b_0)$ carry a Hamiltonian action of $\Diff^{+}(S^1)$. 
In what follows, we will be interested in the action of the subgroup $S^1 \subset \Diff^{+}(S^1)$ given by rigid rotations of the circle.
Expressed in terms of $f \in \widetilde\Diff^{+}(S^1)$, this action takes the form
\begin{equation}
  \RR \ni \theta\colon f(x) \mapsto f(x+\theta), 
  \label{eq:circle_action}
\end{equation}
which descends to an $S^1$-action on the orbit.
Its moment map is given by
\begin{equation}
  \mu_{S^1} = \int_{S^1} b^f(x) dx.
\end{equation}
It is known (see for instance \cite{Dai-Pickrell}) that $\Stab(b_0dx^2) = S^1$ and that $b_0$ is the unique fixed point of the circle action \eqref{eq:circle_action}
on $\OO_{b_0}$.

\section{Darboux coordinates on Virasoro orbits}\label{sec:darboux_coordinates_for_hyperbolic_orbits}
The aim of this Section is to introduce global Darboux coordinates on hyperbolic Virasoro coadjoint orbits which admit a constant representative $b_0 dx^2$ and on the parabolic orbit which corresponds to the vanishing quadratic differential.

\subsection{Darboux charts in finite dimensions}
In this Section, we recall the classical equivariant Darboux theorem in the finite dimensional situation. 

Consider a symplectic manifold $(M^{2n},\omega)$ endowed with a Hamiltonian $S^1$-action. Let $\mu: M \to \mathbb{R}$ be the corresponding moment map,
and let $p \in M$ be an isolated fixed point of the $S^1$-action.
The equivariant Darboux theorem \cite{Weinstein} states that there is a neighborhood of $p$  which is equivariantly symplectomorphic to a neighborhood of $0 \in \CC^n$ endowed with the standard symplectic form 
\begin{equation}  \label{eq:omega_0}
\omega_0 = \frac{i}{2}\sum_j dz_j \wedge d\bar z_j,
\end{equation}
moment map 
\begin{equation}  \label{eq:moment_0}
    \mu = \mu(p) + \frac12 \sum_j k_j \lvert z_j \rvert^2,
\end{equation}
and $S^1$-action 
\begin{equation}  \label{eq:action_0}
    z_j \mapsto e^{i k_j \theta} z_j.
\end{equation}
Here $0 \neq k_j \in \mathbb{Z}$ are the weights of the $S^1$-action at $p$.

In general, there are no estimates on the size of an equivariant Darboux chart.
However, there is a special situation where a rather precise estimate is available. The following theorem is a reformulation of the result by Karshon-Tolman \cite{Karshon-Tolman}:

\begin{theorem}  \label{thm:KT}
Assume that all the weights $k_j$ are positive and that the $S^1$-action is free on $\mu^{-1}(a,b)$, where $a = \mu(p)$. Then, there is an $S^1$-equivariant symplectomorphism from $\mu^{-1}([a, b))$ to the ball $\{(z_1, \dots, z_n); \sum k_j|z_j|^2 < 2(b-a)\}$ equipped with the $S^1$-action \eqref{eq:action_0}, symplectic structure \eqref{eq:omega_0} and moment map \eqref{eq:moment_0}.
A similar result holds when all the weights $k_j$ are negative.
\end{theorem}

Theorem \ref{thm:KT} implies that if all the weights are of the same sign and if $p$ is the only fixed point, then the corresponding equivariant Darboux chart is global (that is, it covers all of $M$).



As an example, consider the unit two-sphere $S^2$ embedded in $\RR^3$  with its South pole $S$ removed and endowed with the standard symplectic form  \mbox{$\omega = dz \wedge d\varphi$}.
Here $z$ denotes the $z$-coordinate on $\RR^3$ and $\varphi$ is the polar angle. 
Consider the $S^1$-action which rotates the sphere around the $z$-axis. 
This action is Hamiltonian with  moment map 
$\mu \colon (x,y,z) \mapsto z$.
The circle action admits a unique fixed point, the North pole, sitting at $z = +1$. 
Around the 
fixed point, the Darboux coordinates are given by
\begin{equation*}
  \begin{split}
    x &= \sqrt{2(1-z)}\sin(\varphi), \\
    y &= \sqrt{2(1-z)}\cos(\varphi),
  \end{split}
\end{equation*}
such that 
\begin{equation*}
  \begin{split}
    \omega &= dx \wedge dy \\
    \mu &= 1 - \frac12 (x^2 + y^2).
  \end{split}
\end{equation*}
This chart is indeed global on $S^2 \backslash \{S\}$.\bigskip

The crucial observation is that for a large class of Virasoro coadjoint orbits the circle action \eqref{eq:circle_action} admits a unique fixed point (see {\em e.g.}\  \cite{Dai-Pickrell}). 
Moreover, there is a class of orbits for which all the weights at this fixed point have the same sign.
This suggests that these orbits may admit a global Darboux chart. Note that the
Karshon-Tolman theorem cannot be directly applied to this situation since 
Virasoro coadjoint orbits are infinite dimensional.

\subsection{Hyperbolic orbits with \texorpdfstring{$b_0<0$}{b0<0}}
In this Section, we construct global equivariant Darboux charts on Virasoro coadjoint orbits passing through constant quadratic differentials $b_0dx^2$ for $b_0 < 0$ and $c > 0$. 

We will need the following notation: for $\alpha >0$, let
\begin{align*}
\tilde B_{\alpha} &= \{ F \colon \mathbb{R} \to \mathbb{R}_+; F'(x) >0, F(x+2\pi) = e^{2\pi \alpha} F(x)\}.
\intertext{and}
\tilde A_{\alpha} &= \{ u \colon \mathbb{R} \to \mathbb{R}; u(x+2\pi) = u(x) + 2\pi \alpha\}
\end{align*}
Furthermore, we introduce two maps 
$r_\alpha$ and $p_\alpha$ and their composition \mbox{$q_\alpha = p_\alpha \circ r_\alpha$}.
The map
\[
r_\alpha \colon \widetilde\Diff^+(S^1) \to \tilde B_{\alpha} 
\]
is given by the formula
$$
r_{\alpha} \colon f \mapsto F(x) = (g_\alpha \circ f)(x) = \frac{e^{\alpha f(x)}}{\alpha},
$$
where $g_\alpha(y)=e^{\alpha y}/{\alpha}$.
The second map,
\[
p_{\alpha} \colon \tilde B_{\alpha} \to \tilde A_{\alpha}
\]
is defined by the formula
$$
p_\alpha \colon F \mapsto u(x) = \log(F'(x)).
$$
It is easy to see that the quasi-periodic properties of $f$ imply the quasi-periodic properties of $F$, and that those of $F$ imply those of $u$. 

\begin{proposition}
  The maps $r_\alpha$ and $p_\alpha$ are bijections.
\end{proposition}

\begin{proof}
 The inverse of the map $r_\alpha$ is given by
 \[
 r_\alpha^{-1}\colon F \mapsto f(x) = \frac{1}{\alpha} \, \log(\alpha F(x)).
\]
 Furthermore, we have $f'(x) = F'(x)/(\alpha F(x))>0$ and
 \[
 f(x+2\pi) = \frac{1}{\alpha} \, \log(\alpha F(x+2\pi)) =
 \frac{1}{\alpha} \, \log\left(\alpha e^{2\pi\alpha}F(x)\right) =
 f(x) + 2\pi.
\]
 The inverse of $p_\alpha$ is given by
\[
 p_\alpha^{-1} \colon u \mapsto F(x) = \frac{1}{e^{2\pi \alpha}-1}\left( e^{2\pi \alpha} \int_0^x e^{u(s)} ds + \int_x^{2\pi} e^{u(s)} ds\right).
\]
Note that $F'(x)=e^{u(x)}>0$ and $F(2\pi) = e^{2\pi \alpha}F(0)$, as required.
\end{proof}

The composition map $q_\alpha$ is also a bijection given by the formula
\[
q_\alpha \colon f \mapsto u(x) = \alpha f(x) + \log(f'(x)).
\]
Its inverse $q_\alpha^{-1}$ is given by an explicit (albeit a bit cumbersome) expression:
\[
q_\alpha^{-1}: u \mapsto  f(x) = \frac{1}{\alpha} \left( \log\left( e^{2\pi \alpha} \int_0^x e^{u(s)}ds + \int_x^{2\pi} e^{u(s)}ds \right) + \log\left( \frac{\alpha}{e^{2\pi \alpha} - 1} \right) \right).
\]

In what follows, we will be interested in the quotient space
\[
\widetilde\Diff^+(S^1)/\mathbb{R} \cong \Diff^+(S^1)/S^1.
\]
Here the action of $\mathbb{R}$ is given by $f(x) \mapsto f(x) + \phi$.
It is convenient to introduce a special notation for  quotient spaces
\[
B_\alpha = \tilde B_{\alpha} / \RR 
\]
under the action $F(x) \mapsto e^{2\pi \alpha \phi} F(x)$, and for quotient spaces
\[
A_\alpha = \tilde A_{\alpha} / \RR 
\]
under the action $u(x) \mapsto u(x) + \alpha \phi$.

It is easy to see that the maps $r_\alpha, p_\alpha$ and $q_\alpha = p_\alpha \circ r_\alpha$ descend to bijections between these quotient spaces. In particular the bijection
$$
q_\alpha: \Diff^+(S^1)/S^1  \to A_\alpha
$$
lands in the affine space $A_\alpha$
modelled on the space of $2\pi$-periodic functions with zero mean:
$$
u(x) = \alpha x + \sum_{n \neq 0} u_n e^{inx},
$$
where $u_{-n}=\bar{u}_n$.

One of the main results of this article is the following theorem:
\begin{theorem}\label{thm:main_theorem}
The map $q_\alpha$ is an equivariant symplectic isomorphism between
the hyperbolic coadjoint orbit with $b_0 = - c\alpha^2/24$ and the affine space
$A_\alpha$ with symplectic form
\begin{equation}  \label{eq:omega_alpha}
\omega = \frac{c}{24} \int_{S^1} (\delta u \wedge \delta u') dx = 
\frac{i\pi c}{6} \sum_{n>0} n \, \delta u_n \wedge \delta \bar{u}_n,
\end{equation}
action of $\widetilde\Diff^+{S^1}$
\begin{equation}  \label{eq:affine_alpha}
h \colon u \mapsto u^h(x) = u(h(x)) + \log(h'(x)),
\end{equation}
and moment map
\begin{equation} \label{eq:moment_alpha}
\mu(u) = \frac{c}{12}\left( u''(x) - \frac{1}{2} \, u'(x)^2\right).
\end{equation}
\end{theorem}

\begin{proof}
First, we show that the map $q_\alpha$ intertwines the actions of $\widetilde\Diff^+(S^1)$ on itself (by right multiplications) and on the space of quasi-periodic functions $u(x)$:
%
\begin{alignat*}{2}
f \circ h &\mapsto&~ &\alpha f(h(x)) + \log((f \circ h)'(x)) \\
&=&~  &\alpha f(h(x)) + \log(f'(h(x))) + \log(h'(x)) \\
&=&~  &u(h(x)) + \log(h'(x)).
\end{alignat*}
These actions descend to actions on the quotient spaces $\widetilde\Diff^+(S^1)/\mathbb{R}\cong \Diff^+(S^1)/S^1$ and $A_\alpha$.

Next, we transfer the moment map and the symplectic form from the Virasoro coadjoint orbit to the space $B_\alpha$. 
 Note that the function $g_\alpha$ has the following properties:  $g'_\alpha(x) > 0$, $S(g_\alpha) = -\alpha^2/2 =12 b_0/c$, and
$g''_\alpha/g'_\alpha=\alpha$ is constant. 

By the properties of Schwarzian derivative, we have
\[
\mu = b_0 f'(x)^2 + \frac{c}{12} S(f) = \frac{c}{12}\left( S(g_\alpha) \circ f + S(f) \right) = \frac{c}{12} \, S(F).
\]
Furthermore,
%
\begin{equation}
    \begin{split}
            \omega &=   \frac{c}{24} \int_{S^1} \left( \delta\log F' \wedge \delta\log(F')'  \right)dx - \left.\frac{g''_\alpha(f(x))}{g'_\alpha(f(x))} \frac{\delta f \wedge \delta f'}{f'}\right\rvert_0^{2\pi} \\
    &= \frac{c}{24} \int_{S^1} \left( \delta\log F' \wedge \delta\log(F')'  \right)dx. 
    \end{split}
    \label{eq:omega_in_terms_of_F}
\end{equation}

Finally, we use the map $p_\alpha$ to transfer the moment maps and the symplectic form to the affine space $A_\alpha$:
\begin{equation}
\mu = \frac{c}{12}\left( u''(x) - \frac{1}{2} \, u'(x)^2\right), \hskip 0.3cm
\omega = \frac{c}{24} \int_{S^1} (\delta u \wedge \delta u') dx.
\label{eq:mu_and_omega_hyp}
\end{equation}
This completes the proof.
\end{proof}

Theorem \ref{thm:main_theorem} provides global $S^1$-equivariant Darboux charts on all hyperbolic Virasoro coadjoint orbits $\OO_{b_0}$ passing through constant quadratic differentials $b_0 dx^2$ with $b_0<0$. The $S^1$-action in this Darboux chart is given by
$$
\theta: u(x) \mapsto u(x+\theta).
$$
In terms of Fourier components $u_n$, we have $u_n \mapsto e^{in\theta} u_n$, and the weights of the action at the fixed point $u_n=0$ are given by $n=1, 2, \dots$. The moment map of the $S^1$-action is given by
\begin{equation}
\mu_{S^1} = \int_{S^1} \mu(u) dx = - \frac{c}{24} \int_{S^1} u'(x)^2 dx =
-\frac{c}{12}\left( 2\pi \alpha^2 + 4\pi \sum_{n>0} n^2 |u_n|^2\right).
\label{eq:mu_S1_hyp}
\end{equation}

Furthermore, the group $\Diff^+(S^1)$ acts on the symplectic affine space $A_\alpha$  by affine Hamiltonian transformations \eqref{eq:affine_alpha} with moment map \eqref{eq:moment_alpha} given by quadratic expressions in $u(x)$.

\subsection{The parabolic orbit \texorpdfstring{$b_0=0$}{b0=0}}
In this Section, we consider the special case of the parabolic Virasoro orbit which corresponds to the vanishing quadratic differential.

We define a map
\[
q_0 \colon \widetilde\Diff^+(S^1) \to \tilde A_0 = \{ u: \mathbb{R} \to \mathbb{R}; u(x+2\pi)=u(x)\}
\]
given by the formula
\[
q_0 \colon f \mapsto u(x)=\log(f'(x)).
\]
It descends to a map (that we again denote by $q_0$)
\[
q_0 \colon \widetilde\Diff^+(S^1)/\mathbb{R} \to A_0 = \{ u: \mathbb{R} \to \mathbb{R}; u(x+2\pi)=u(x)\}/\mathbb{R}.
\]
Here the the $\mathbb{R}$-actions are as follows: $f(x) \mapsto f(x) + \phi, u(x) \mapsto u(x) + \psi$.

\begin{proposition}
  The map $q_0: \widetilde\Diff^+(S^1)/\mathbb{R} \to A_0$ is a bijection.
\end{proposition}

\begin{proof}
  The inverse map is given by the following formula:
  \[
  q_0^{-1}: u \mapsto f(x) = f(0) + \frac{2\pi}{\int_0^{2\pi} e^{u(s)} ds} \int_0^x e^{u(t)} dt.
  \]
  Note $q_0^{-1}$ is a well defined map on $A_0$ since $f(x)$ doesn't change under shifts $u(x) \mapsto u(x) + \psi$. It defines a class of diffeomorphisms $f(x)$ modulo transformations $f(x) \mapsto f(x) + \phi$ which amount to the choice of $f(0)$.
\end{proof}

Similar to Theorem \ref{thm:main_theorem}, we obtain the following result:
\begin{theorem}  \label{thm:parabolic}
The map $q_0$ is an equivariant symplectomorphism between the parabolic coadjoint orbit corresponding to $b_0=0$ and the vector space $A_0$ with the symplectic form, the action of $\widetilde\Diff^+(S^1)$ and the moment map given by equations \eqref{eq:omega_alpha},\eqref{eq:affine_alpha} and \eqref{eq:moment_alpha}.
\end{theorem}

\begin{proof}
  For the action, we compute
  \[
  u^h(x) = \log((f\circ h)'(x)) = \log(f'(h(x))) + \log(h'(x)) = u(h(x)) + \log(h'(x)).
  \]
  
  The moment map is given by
  \[
  \mu = \frac{c}{12} S(f) = \frac{c}{12}\left( u''(x) - \frac{1}{2} u'(x)^2\right) .
  \]
  
  Finally, for the symplectic form we obtain,
  \[
  \omega = \frac{c}{24} \int_{S^1} (\delta \log(f') \wedge \delta(\log(f'))') dx =
  \frac{c}{24} \int_{S^1} (\delta u \wedge \delta u') dx.
  \]
\end{proof}

Theorem \ref{thm:parabolic} provides a global $S^1$-equivariant Darboux chart on the coadjoint orbit passing through the vanishing quadratic differential. Again, we actually obtain a stronger result: the action of the diffeomorphism group $\Diff^+(S^1)$ on $A_0$ is by  affine Hamiltonian transformations with moment map given by a quadratic polynomial in $u(x)$.

\section{The Teichm\"uller orbit}\label{sec:problems_of_the_teichmueller_orbit}  

Recall that generic elliptic orbits corresponding to $b_0>0$ and exceptional orbits corresponding to $b_0=ck^2/24$ each possess a unique fixed point of the $S^1$-action of rigid rotations $x \mapsto x + \theta$. Therefore, one can conjecture that they also admit global equivariant Darboux charts:

\vskip 0.2cm

{\bf Conjecture:} Virasoro coadjoint orbits corresponding to $b_0>0$ admit global $S^1$-equivariant Darboux charts.

\vskip 0.2cm

To the best of our knowledge, this conjecture remains out of reach. However, techniques of the last Section can be applied to the case of exceptional orbits, and they give insight into their geometry. In this Section, we focus our attention on the Techm\"uller orbit which corresponds to $k=1$. 
By abuse of notation, we will denote for any $f(x) \in \widetilde\Diff^+(S^1)$ its class $[f(x)] \in \Diff^+(S^1)$ also by $f(x)$. \bigskip

Recall that the Teichm\"uller orbit is the unique Virasoro coadjoint orbit whose stabiliser is isomorphic to 
the group ${\rm PSL}(2, \mathbb{R})$:
\[
\OO_{\rm Teich} \cong {\rm Diff}^+(S^1)/{\rm PSL}(2, \mathbb{R}).
\]
The corresponding action is given by M\"obius transformations
\[
F(x) \mapsto \frac{a F(x) +b}{cF(x) +d},
\]
where
\[
F(x) = 2 \tan\left( \frac{f(x)}{2} \right).
\]

Following the strategy of the previous Section, we define
\[
v(x) = \log(F'(x)) = - \log(\cos^2(f(x)/2)) + \log(f'(x)).
\]
Under the action of diffeomorphisms, $v(x)$ transforms 
according to \eqref{eq:affine_alpha}. Indeed,
\[
v^h(x) = \log(F(h(x))') = \log(F'(h(x)) + \log(h'(x)) = v(h(x)) + \log(h'(x)).
\]
However, $v(x)$ has a singularity at a point $y$ 
where $f(y)=\pi$. 
In more detail, it belongs to the space
\[
C =\{ v\colon \mathbb{R} \to \mathbb{R}; v(x+2 \pi) = v(x), 
v(x) = - \log(\sin^2((x-y)/2)) + u(x)\}, 
\]
where $y \in \mathbb{R}$ is some point on the real axis, and $u(x)$ is smooth. The singularity of $v(x)$ at $x=y$ is of the form $v(x) \sim - 2\log(|x-y|)$. However, it is more convenient to use the function $\sin^2((x-y)/2)$ since it is $2\pi$-periodic and has the same behavior at zero. 

The space $C$ is a vector bundle over the circle with projection map
\begin{equation}
v(x) \mapsto y ~ {\rm mod} ~ 2 \pi \mathbb{Z},
\label{eq:v}
\end{equation}
and with fibers $C_y$ isomorphic to spaces of smooth periodic functions (hence, fibers are independent of $y$). It is easy to check that the action \eqref{eq:affine_alpha} is well defined on $C$.

We define a subbundle $D\subset C$ as follows:
\[
D=\left\{ v \in C;\; u'(y)=0, \int_{S^1} \frac{e^{u(x)} - e^{u(y)}}{\sin^2\left((x-y)/2\right)} \, dx=0\right\}.
\]
Note that for a smooth function $u(x)$ with $u'(y)=0$ the expression under the integral is smooth at $x=y$, and the integral is well defined. The fiber of $D_y$ is singled out  by the two equations. While the equation $u'(y)=0$ is linear, the integral condition is highly non-linear, and therefore the fiber $D_y$ is not a vector space. 

It will be convenient to introduce one more space as follows: consider configurations of two points on the circle
\[
{\rm Conf}_2(S_1) = \{ (y,z) \in S^1 \times S^1; y \neq z\}.
\]
It has a natural projection to the first point: $(y,z) \mapsto y$. Denote by $E$ the pull-back of the bundle $D$ under this projection. That is, for all $z \neq y$ the fiber $E_{y,z} = D_y$. Definitions of the spaces $D$ and $E$ are justified by the following result:
\begin{theorem}  \label{thm:f_teich}
For $f \in \Diff^+(S^1)$, the map 
\begin{equation}
  u \colon f \mapsto u(x) = \log(\sin^2((x-y)/2)) - \log(\cos^2(f(x)/2)) + \log(f'(x)),
  \label{eq:u}
\end{equation}
where $f(y)=\pi$, takes values in $D$. Furthermore, the map $f \mapsto (u(x), y, z)$ with
$f(z)=0$ defines a bijection $\Diff^+(S^1) \to E$.
\end{theorem}
\begin{remark}
 The change of variable $v(x)$ in \eqref{eq:v} has been previously used in \cite{Altland} for the analysis of correlation functions of the SYK model. 
 Due to the singularity of $v(x)$, a cut-off is needed to regularize the integral constraint.
 In \eqref{eq:u}, the singularity is subtracted from the beginning.
\end{remark}
\begin{proof}
  The function $u(x)$ admits an expansion at $x=y$. Its first terms read:
  %
\begin{align*}
    u(x) &=  \log\left(\frac{(x-y)^2}{4} \right) - \log\left( \frac{f'(y)^2 (x-y)^2}{4}\right) - \frac{f''(y)}{f'(y)} \, (x-y) \\
  &+  \log(f'(y)) + \frac{f''(y)}{f'(y)} \, (x-y) + O((x-y)^2) \\
  &=  - \log(f'(y)) + O((x-y)^2).
\end{align*}
  That is, the function $u(x)$ is $2\pi$-periodic and smooth at $x=y$. We also observe that  $u(y)=-\log(f'(y))$ and $u'(y)=0$.
  
  Next, consider the function
  \[
  \Phi(x) = 2\tan\left(\frac{f(x)}{2}\right) + \frac{2}{f'(y)} \, {\rm cotan}\left(\frac{x-y}{2}\right).
  \] 
  At $x=y$, it admits a Laurent series expansion whose first terms are given by
  \[
  \Phi(x) = - \frac{4}{f'(y)} \, \frac{1}{x-y} + 2 \frac{f''(y)}{f'(y)} + \frac{4}{f'(y)} \, \frac{1}{x-y} + O(x-y) = 2 \frac{f''(y)}{f'(y)} + O(x-y).
  \]
  Hence, $\Phi(x)$ is smooth and $2\pi$-periodic. This implies
  \begin{align*}
  0 &= \int_{S^1} \Phi'(x) dx = \int_{S^1} \left( \frac{f'(x)}{\cos^2(f(x)/2)} - \frac{e^{u(y)}}{\sin^2((x-y)/2)}\right)\, dx \\
  &=
  \int_{S^1} \frac{e^{u(x)} - e^{u(y)}}{\sin^2((x-y)/2)} \, dx,
  \end{align*}
  and we conclude that $u(x)$ indeed belongs to $D$.
  
  The inverse map $E \to \Diff^+(S^1)$ is given by the formula
  \[
  (u(x), y, z) \mapsto f(x) = 2 \arctan\left( \frac{1}{2} \int_z^x \frac{e^{u(s)}}{\sin^2((s-y)/2)} \, ds \right). 
  \]
  It is easy to check that if $(u(x), y, z) \in E$, then this expression defines an element $f\in  \Diff^+(S^1)$ with $f(z)=0$ and $f(y)=\pi$. This completes the proof. 
\end{proof}

Theorem \ref{thm:f_teich} gives rise to the following description of the Teichm\"uller orbit:
\begin{proposition}
\begin{equation*}
    \OO_{\rm Teich} \cong 
    \left\{ u\colon \mathbb{R} \to \mathbb{R}; u(x+2\pi) = u(x), u(\pi) =u'(\pi)=0, \int_{S^1} \frac{e^{u(x)}-1}{\cos^2(x/2)} \, dx =0\right\}.
\end{equation*}
\end{proposition}

\begin{proof}
  Since 
  $\Diff^+(S^1) \cong E$, the space $E$ inherits the free ${\rm PSL}(2, \mathbb{R})$ action. This action admits a section: each ${\rm PSL}(2, \mathbb{R})$ orbit has exactly one point with $f(0)=0, f(\pi)=\pi$ and $f'(\pi)=1$. This implies that $y=\pi, z=0$ and $u(\pi)=-\log(f'(\pi))=0$. Hence, the Teichm\"uller orbit is identified with the fiber $E_{\pi, 0}$ with an extra condition of $u(\pi)=0$.
\end{proof}

We now describe the moment map and the symplectic form on the Teichm\"uller orbit in terms of this model.

\begin{theorem}  \label{thm:u_teich}
The moment map and the symplectic form on the Teichm\"uller orbit are given by
\begin{equation}
    \mu = \frac{c}{12}\left( u''(x) - \frac{1}{2} u'(x)^2 - u'(x) \tan\left(\frac{x}{2}\right) + \frac{1}{2} \right),
\end{equation}
\begin{equation}  \label{eq:omega_teich}
    \omega = \frac{c}{24} \int_{S^1} (\delta u \wedge \delta u') \, dx.
\end{equation}
\end{theorem}

\begin{proof}
  For the moment map, we compute
  \[
  \mu = \frac{c}{12}\left(\frac{1}{2} f'(x)^2 + S(f)\right) = 
  \frac{c}{12}\left( (S(g) \circ f) f'(x)^2 + S(f) \right) =
  \frac{c}{12} \, S(F)
  \]
  where $g(y) = 2 \tan(y/2)$ and  $F(x) = g(f(x))$. Here we have used the composition formula for the Schwarzian derivative and the fact that $S(g)=1/2$. Then,
  since
  \[
  F'(x) = \frac{e^{u(x)}}{\cos^2(x/2)},
  \]
  we obtain
  \[
  S(F) = u''(x) - \frac{1}{2}u'(x)^2 - u'(x) \tan\left(\frac{x}{2}\right) + \frac{1}{2},
  \]
  as required. Note that the expression $u'(x) \tan(x/2)$ is smooth at $x=\pi$ since $u'(\pi)=0$.
  
  For the calculation of $\omega$, we use equation \eqref{eq:omega_in_terms_of_F} with the periodic function $g(y) = \tan(y/2)$. Observe  that $\delta \log(F') = \delta u$ since the factor $\cos^2(x/2)$ does not contribute in the de Rham differential. This  completes the proof. 
\end{proof}

Theorem \ref{thm:u_teich} implies the following expression for the moment of rigid $S^1$-rotations:
\begin{equation} \label{eq:moment_teich}
\mu_{S^1} = \frac{c}{12} \, \int_{S^1} \left( u'(x) \tan\left(\frac{x}{2}\right) -\frac{1}{2} u'(x)^2\right) + \frac{\pi c}{12}.
\end{equation}
Observe that in this model the symplectic form \eqref{eq:omega_teich} is constant, and the moment map \eqref{eq:moment_teich} is quadratic. However, this is not an equivariant Darboux chart because the function $u(x)$ satisfies the non-linear integral constraint. The Conjecture in the beginning of the Section states that one can find a transformation which would remedy this problem and obtain a model of the Teichm\"uller orbit on a vector space with constant symplectic structure and quadratic moment map.

\section{Partition function and correlation functions of bilocal operators}\label{sec:partition_function_and_correlation_functions_of_bilocal_operators}
In this section we compute the partition function and correlation functions of bilocal operators in  Darboux charts on hyperbolic Virasoro coadjoint orbits. While the partition function readily recovers known results \cite{Alekseev_Shatashvili_char_orbits_DH, Stanford_Witten}, the Darboux chart makes the exact computation of correlation functions feasible. 

\subsection{Partition function}
The partition function of a Schwarzian theory is a path integral of the form
\begin{equation}
  Z(t) = \int  d\lambda(f)e^{t\mu_{S^1}(f)},
  \label{eq:Z}
\end{equation}
where $\mu_{S^1}(f)$ is the moment map for rigid $S^1$-rotations, $t$ is a parameter, and 
$d\lambda(f)$ stands for some (usually ill defined) measure on ${\rm Diff}^+(S^1)$.

In the case of hyperbolic coadjoint orbits, the global equivariant Darboux chart allows to define the path integral \eqref{eq:Z} as a Gaussian integral. Indeed, recall that for
\[
u(x) = \alpha f(x) + \log f'(x),
\]
with $\alpha^2 = - 24b_0/c$, the symplectic form $\omega$ and the moment map $\mu_{S^1}$ of rigid rotations are of the following form, {\em c.f.}\ \eqref{eq:mu_and_omega_hyp} and \eqref{eq:mu_S1_hyp}: 
\begin{align*}
  \omega &= \frac{c}{24} \int_{0}^{2\pi} \left( \delta u(x) \wedge \delta u'(x) \right) dx, \\
  \mu_{S^1} &= - \frac{c}{24} \int_0^{2\pi} u'^2(x) dx. 
\end{align*}
Here $u(x+2\pi)=u(x)+2\pi \alpha$.

Inspired by the finite dimensional situation, one can argue that the integration measure $d\lambda(f)$ is the formal Liouville measure
$$
d\lambda(f) = \mathcal{D} u \cdot {\rm Pf}_\zeta(\omega),
$$
where $\mathcal{D} u$ is the formal Lebesgues measure (which is ill defined), and 
${\rm Pf}_{\zeta}(\omega)$ is  the $\zeta$-regularized Pfaffian of the symplectic form
\begin{equation*}
  {\rm Pf}_{\zeta}(\omega) = {\rm Pf}_{\zeta}\left( \frac{c}{12} \del_x \right) 
  = \left( \frac{24\pi}{c} \right)^{1/2}.
\end{equation*}
While different elements in the path integral \eqref{eq:Z} are ill defined, the whole expression is a well defined Gaussian integral. Indeed, we obtain
\begin{equation}
  Z(t) = {\rm Pf}_{\zeta}(\omega)\, \int \mathcal{D} u \, e^{t\mu_{S^1}(u)} = \left( \frac{24\pi}{c} \right)^{1/2} \frac{e^{t\mu_{S^1}(u_{cl})}}{\sqrt{\det_{\zeta}(-\frac{ct}{12}\Delta)}}
\end{equation}
where $\Delta$ denotes the Laplacian on the circle and $u_{cl}$ is a solution of the classical equations of motion $u''(x) = 0$ satisfying the quasi-periodicity condition $u(x + 2\pi) = u(x) + 2\pi\alpha$, {\em i.e.}\ 
\begin{equation*}
  u_{cl}(x) = \alpha x + u_0.
\end{equation*}
The expression  $\det_{\zeta}$ stands the $\zeta$-regularized determinant.
In the case at hand it is given by
\begin{equation*}
  \det\nolimits_{\zeta}(-\frac{ct}{12}\Delta) = \frac{48\pi^2}{ct}.
\end{equation*}
Putting things together, we conclude that
\begin{equation}
    Z(t) 
    = \left( \frac{t}{2\pi}\right)^{1/2} \exp\left(- \frac{c t \alpha^2 \pi }{12}\right) 
    = \left( \frac{t}{2\pi}\right)^{1/2} e^{ 2 \pi b_0 t }
    \label{eq:Z_Gaussian}
\end{equation}
This recovers the results of  \cite{Alekseev_Shatashvili_char_orbits_DH,Stanford_Witten}.

\subsection{Correlator of a single bilocal operator}
In the literature, 
one often studies ${\rm PSL}(2,\RR)$-invariant bilocal quantities $\OO(x_1,x_2)$ which in terms of $F(x) = e^{\alpha f(x)}/\alpha$ 
are given by
\begin{equation}
  \OO(x_1,x_2) = \frac{F(x_2) - F(x_1)}{\sqrt{F'(x_2) F'(x_1)} },
\end{equation}
where we assume that $x_2 > x_1$.
In terms of the diffeomorphism $f(x)$, one finds the expression
\begin{equation*}
  \OO(x_1,x_2) = \frac{2}{\alpha} \frac{\sinh\left( \frac{f(x_2) - f(x_1)}{2} \right)}{\sqrt{f'(x_2)f'(x_1)}},
\end{equation*}
and in Darboux coordinates $u(x) = \log F'(x)$, we obtain 
\begin{equation}
  \OO(x_1,x_2) = e^{- (u(x_1) + u(x_2) )/2 } \int_{x_1}^{x_2} e^{u(s)} ds.
\end{equation}

In this Section, we will compute the correlation function (the Gaussian expectation value) of the following form:
\begin{equation}
  \begin{split}
  \left\langle \OO(x_1,x_2) \right\rangle &= \frac{1}{Z(t)} \int d\lambda(f) \OO(x_1,x_2) e^{t\mu_{S^1}(f)} \\
  &= \frac{{\rm Pf}_\zeta(\omega)}{Z(t)} \int_{x_1}^{x_2} ds \int \mathcal{D}u\  e^{\mathcal S(u,s,t)}.
  \end{split}
  \label{eq:1_insertion}
\end{equation}
Here, 
\begin{equation*}
 \mathcal S(u,s,t) = -\frac{ct}{24} \int_{0}^{2\pi} u'^2(x) dx - \frac12 u(x_1) - \frac12 u(x_2) + u(s)
\end{equation*}
is a quadratic expression in $u$ (in the physics language, the effective action).\footnote{It is standard to denote the action functional by the letter $\mathcal S$; the action $\mathcal S(u,s,t)$ should not be confused with the Schwarzian derivatives $S(f)$ and $S(F)$.} 
We observe that the correlator \eqref{eq:1_insertion} is again a Gaussian integral and hence exactly solvable
\begin{equation}
\begin{split}
  \left\langle \OO(x_1,x_2) \right\rangle 
  &= \frac{{\rm Pf}_{\zeta}(\omega)}{Z(t)} \int_{x_1}^{x_2}ds\ \frac{e^{\mathcal S(u_{cl}, s, t)}}{\sqrt{\det_{\zeta}(- \frac{ct}{12}\Delta)}} \\
  &= e^{-2\pi b_0 t} \int_{x_1}^{x_2} ds\ e^{\mathcal S(u_{cl},s, t)},
  \end{split}
\end{equation}
where $u_{cl}$ denotes the critical point (in this case, the maximiser) of the quadratic functional 
$\mathcal S(u,s,t)$.
The function $u_{cl}$ is continuous, and it is  subject to the following conditions:
\begin{enumerate}[(i)]
  \item away from $x_1, x_2$ and $s$, $u_{cl}(x) = mx + n$ is a linear function
  \item $u_{cl}$ satisfies the quasi-periodicity condition $u_{cl}(x + 2\pi) = u_{cl}(x) + 2\pi\alpha$
  \item $u'_{cl}$ is discontinuous  
    \begin{equation}
      \begin{cases} u_{cl}'(x_1 + \varepsilon) - u_{cl}'(x_1 - \varepsilon) = \phantom{-}\frac{6}{c t} \\ u_{cl}'(x_2 + \varepsilon) - u_{cl}'(x_2 - \varepsilon) =  \phantom{-}\frac{6}{c t} \\ u_{cl}'(s\phantom{_2} + \varepsilon) - u_{cl}'(s\phantom{_2} - \varepsilon) = - \frac{12}{c t}  \end{cases}
      \label{eq:bc}
    \end{equation}
\end{enumerate}
This problem is solved by a piecewise linear Ansatz
\begin{equation*}
  \begin{split}
    u_{cl}(x) = 
    \begin{cases}
      mx + n \qquad & x \leq x_1 \\
      mx_1 + n + a(x-x_1) \qquad & x_1 < x \leq s \\
      mx_1 + n + a(s-x_1) + b(x - s) \qquad & s < x \leq x_2 \\
      mx_1 + n + a(s-x_1) + b(x_2 - s) + c (x-x_2) \qquad &x_2 \leq 2\pi
    \end{cases}
  \end{split}
\end{equation*}
where we assume that $x_1 < s < x_2$.
From \eqref{eq:bc}, we find a relation among the slopes
\begin{equation*}
  a = m + \frac{6}{c t}, \qquad b = m - \frac{6}{c t}, \qquad c = m. 
\end{equation*}
The quasi-periodicity condition then yields an expression for $m$ as a function of $s$:
\begin{equation*}
  m(s) = \alpha - \frac{3}{c t \pi}\left( 2s - x_1 - x_2 \right) . 
\end{equation*}
Putting things together, we obtain 
\begin{equation}
  \mathcal S(u_{cl},s,t) = \frac{3}{2c t}(x_2 - x_1) - \frac{3}{4\pi c t} \left(2s - x_1 - x_2 - \frac{ct\pi}{3}\alpha\right)^2,
\end{equation}
and 
\begin{equation}
  \begin{split}
    \left\langle \OO(x_1,x_2) \right\rangle &= e^{- 2\pi b_0 t} \int_{x_1}^{x_2}ds\ e^{\mathcal S(u_{cl}, s, t)} \\
    &= \pi e^{-2\pi b_0 t + \frac{3 x_{21}}{2c t}} \sqrt{\frac{c t}{12}} \Big( G(x_{21}) - G(-x_{21}) \Big)
\end{split}
\label{eq:bilocal}
\end{equation}
where we set $x_{21} = x_2 - x_1$ and 
\begin{equation*}
  G(x) = \erf\left( \sqrt{-2\pi b_0} + x \sqrt{\frac{3}{4\pi c t}} \right).
\end{equation*}

\subsubsection{Classical vs quantum calculus} 
In this Section, we use an example of the correlator $\langle \OO(x_1,x_2) \rangle$
to illustrate the difference between the standard (classical) calculus, and calculus under the Gaussian path integral.

In more detail, for $f(x)$ a diffeomorphism of the circle the bilocal quantity $\OO(x_1,x_2)$ is a smooth function in two variables.
In particular, it admits a Taylor expansion $\OO(x_1,x_2)$ near the diagonal $x_1 = x_2$. 
As we will see, this no longer holds under the Gaussian path integral.\medskip

Recall that in terms of  $F = e^{\alpha f(x)} / \alpha$ the moment map and the bilocal observable  take the form
\begin{equation}
    \mu_{S^1}(F) = \frac{c}{12}\int_{0}^{2\pi} S(F) dx,\qquad \OO(x_1, x_2) = \frac{F(x_2) - F(x_1)}{\sqrt{F'(x_2) F'(x_1)} }.
\end{equation}

We will be interested in the Taylor expansion in $\varepsilon$ of the following quantity:
\begin{equation}
  \left\langle \int_{0}^{2\pi}  \OO(x, x + \varepsilon) dx \right\rangle.
\end{equation}
First, we perform the Taylor expansion inside the correlator,
\begin{equation}
  \OO(x, x + \varepsilon) = \varepsilon - \frac{\varepsilon^3}{12}S(F) + \OO\left( \varepsilon^4 \right)
\end{equation}
and substitute it in the path integral: 
\begin{equation}
  \begin{split}
    \left\langle \int_{0}^{2\pi}  \OO(x, x + \varepsilon ) dx \right\rangle &= 2\pi\varepsilon - \frac{\varepsilon^3}{c} \left\langle \frac{c}{12}\int_0^{2\pi} S(F)(x) dx \right\rangle + \OO(\varepsilon^4) \\
   &= 2\pi\varepsilon - \frac{\varepsilon^3}{c} \left\langle \mu_{S^1}(F) \right\rangle + \OO(\varepsilon^4).
   \end{split}
\end{equation}
One can compute the expectation $\left\langle \mu_{S^1}(F) \right\rangle$ as follows:
$$
\left\langle \mu_{S^1}(F) \right\rangle = \frac{d \log(Z(t))}{dt} = 2\pi b_0 + \frac{1}{2t}.
$$
Hence,  
\begin{equation}  
  \left\langle \int_{0}^{2\pi}  \OO(x, x + \varepsilon ) dx \right\rangle = 2\pi\varepsilon - \frac{\varepsilon^3}{c} \left( 2\pi b_0 + \frac{1}{2t} \right) + \OO(\varepsilon^4). 
  \label{eq:Taylor}
\end{equation}
On the other hand,  \eqref{eq:bilocal} 
gives rise to an exact result: 
\begin{equation}
  \begin{split}
    \left\langle \int_0^{2\pi} \OO\left( x, x+\varepsilon \right) dx \right\rangle &= 2\pi^2 e^{-2\pi b_0 + \frac{3\varepsilon}{2ct}} \sqrt{\frac{c t}{12}} \Big( G(\varepsilon) - G(-\varepsilon) \Big) \\
  &= e^{\frac{3 \varepsilon}{2 c t}} \left( 2\pi \varepsilon - \frac{\varepsilon^3}{c}\left( 2\pi b_0 + \frac{1}{2t} \right)\right) + \OO(\varepsilon^4).
\end{split}
\label{eq:exact}
\end{equation}
While the leading degree one terms in \eqref{eq:Taylor} and \eqref{eq:exact} are the same, the subleading terms (starting from $\varepsilon^2$ terms) differ. 
The reason is that classical paths $u_{cl}$ are  not smooth: their first derivatives jump (see \eqref{eq:bc}). Therefore, the Taylor expansion before Gaussian integration is not valid, in general. 

\subsection{Correlators of two bilocals}
Let us study the insertion of two bilocal operators.
In calculating the correlator 
\begin{equation}
  \left\langle \OO(x_1,x_2)\OO(x_3,x_4) \right\rangle 
\end{equation}
we have to distinguish two cases\medskip

\begin{enumerate}[(A)]
  \item time-ordered: \hspace{1.1cm}$0 \leq x_1 < x_2 < x_3 < x_4 \leq 2\pi$, 
  \item out-of-time-ordered: $0 \leq x_1 < x_3 < x_2 < x_4 \leq 2\pi$.
\end{enumerate}\medskip

It is instructive to visualize the two cases as the two possible scenarios of two intervals embedded into the circle: either they are disjoint, or they overlap, {\em c.f.}\ Figure \ref{fig:TO_OTO}. 
\begin{figure}[htb]
  \centering
  \begin{tikzpicture}
    \draw (0,0) circle (1cm);

    \node[draw,circle,fill=black,scale=.3] (x3) at (45:1) {};
    \node[draw,circle,fill=black,scale=.3] (x4) at (135:1) {};
    \node[draw,circle,fill=black,scale=.3] (x1) at (225:1) {};
    \node[draw,circle,fill=black,scale=.3] (x2) at (315:1) {};


    \node[left] at (x1) {$x_1$};
    \node[right] at (x2) {$x_2$};
    \node[right] at (x3) {$x_3$};
    \node[left] at (x4) {$x_4$};

    \draw (x1) edge[bend left] (x2);
    \draw (x3) edge[bend left] (x4);
  \end{tikzpicture}
  \hspace{2cm}
  \begin{tikzpicture}
    \draw (0,0) circle (1cm);

    \node[draw,circle,fill=black,scale=.3] (x2) at (45:1) {};
    \node[draw,circle,fill=black,scale=.3] (x4) at (135:1) {};
    \node[draw,circle,fill=black,scale=.3] (x1) at (225:1) {};
    \node[draw,circle,fill=black,scale=.3] (x3) at (315:1) {};


    \node[left] at (x1) {$x_1$};
    \node[right] at (x2) {$x_2$};
    \node[right] at (x3) {$x_3$};
    \node[left] at (x4) {$x_4$};

    \draw (x1) -- (x2);
    \draw (x3) -- (x4);
  \end{tikzpicture}
  \caption{Left: time-ordered configuration; right: out-of-time-ordered configuration. The circle is oriented counterclockwise.}
  \label{fig:TO_OTO}
\end{figure}
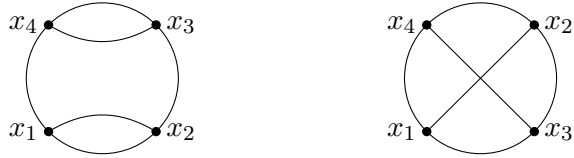

\subsection{Time-ordered correlators}
We start by considering the time-ordered case.
The calculation is analogous to the one of the correlator $\left\langle \OO(x_1,x_2) \right\rangle$ in the previous section.
We recall that in the Darboux coordinate $u(x)$, a bilocal operator takes the form 
\begin{equation*}
  \OO(x_1,x_2) = e^{-(u(x_1) + u(x_2))/2} \int_{x_1}^{x_2} e^{u(s)}ds.
\end{equation*}
Then 
\begin{equation}
  \begin{split}
  \left\langle \OO(x_1,x_2)\OO(x_3,x_4) \right\rangle &= \frac{{\rm Pf}_\zeta(\omega)}{Z(t)} \int \mathcal{D} u\ \OO(x_1,x_2) \OO(x_3,x_4) e^{t\mu_{S^1}(u)} \\
  &= e^{-2\pi b_0 t} \int_{x_3}^{x_4} d\sigma \int_{x_1}^{x_2} d\tau\ e^{\mathcal S(u_{cl},\sigma, \tau, t)}
  \end{split}
\end{equation}
where $u_{cl}$ denotes the critical point of the functional 
\begin{equation}
  \mathcal S(u, \sigma, \tau, t) = - \frac{ct}{24} \int_0^{2\pi} u'^2(x) dx - \frac12\sum_{i=1}^4 u(x_i) + u(\sigma) + u(\tau).
  \label{eq:mu_eff}
\end{equation}
As before, we make a piecewise linear Ansatz for $u_{cl}$ subject to conditions\medskip

\begin{enumerate}[(i)]
  \item away from $x_i$, $\sigma$ and $\tau$, $u_{cl}(x) = mx + n$ is a linear function
  \item $u_{cl}$ satisfies the quasi-periodicity condition $u_{cl}(x + 2\pi) = u_{cl}(x) + 2\pi\alpha$
  \item $u_{cl}$ is continuous and satisfies (for all $i=1,\dots,4$)
    \begin{equation}
      \begin{split}
      \begin{cases}
	u_{cl}'(x_i + \varepsilon) - u_{cl}'(x_i - \varepsilon) &= \phantom{-}\frac{6}{ct} \\ 
	u_{cl}'(\tau + \varepsilon) - u_{cl}'(\tau - \varepsilon) &= -\frac{12}{ct} \\ 
	u_{cl}'(\sigma + \varepsilon) - u_{cl}'(\sigma - \varepsilon) &= -\frac{12}{ct}
      \end{cases}
    \end{split}
    \end{equation}
\end{enumerate}\medskip

\noindent A lengthy calculation shows that
\begin{equation}
  \mathcal S(u_{cl}, \sigma, \tau, t) = \frac{3}{2 c t} \left( x_{43} + x_{21} \right) - \frac{3}{4\pi ct} \left( 2\sigma + 2\tau - \sum_{i=1}^4 x_i - \frac{ct\pi}{3} \alpha \right)^2,
\end{equation}
where $x_{ij} = x_i - x_j$.
Finally, the time-ordered correlator turns out to be 
\begin{equation}
    \left\langle \OO(x_1,x_2)\OO(x_3,x_4) \right\rangle = \frac{\pi}{2} e^{-2\pi b_0 t + \frac{3}{2ct}(x_{43} + x_{21})} \sqrt{\frac{c t}{12}} \Big[  E(x_{21},x_{43}) + G(x_{21},x_{43}) \Big]
\end{equation}
where 
\begin{equation*}
  E(x,y) = 4 \sqrt{\frac{ct}{12}} \sum_{\varepsilon_1,\varepsilon_2 = \pm 1} \varepsilon_1 \varepsilon_2 \exp\left( -\frac{3}{4\pi ct} \left(\varepsilon_1 x + \varepsilon_2 y - \frac{c\pi t}{3} \alpha\right)^2  \right)
\end{equation*}
and
\begin{equation*}
  G(x,y) = \hspace{-1em}\sum_{\varepsilon_1, \varepsilon_2 = \pm 1} \hspace{-.5em} \varepsilon_1 \varepsilon_2 \erf\left( \sqrt{-2\pi b_0 t} - (\varepsilon_1 x + \varepsilon_2 y)\sqrt{\frac{3}{4\pi c t}} \right) \left( \frac{\pi ct}{3} \alpha - \varepsilon_1 x - \varepsilon_2 y \right)
\end{equation*}

\subsection{Out-of-time-ordered correlators}
For out-of-time-ordered correlators, we have to distinguish the five cases shown in Figure \ref{fig:possible_OTO_configuration},
\begin{figure}[htb]
  \centering
  \begin{subfigure}[b]{0.4\textwidth}
    \centering
  \begin{tikzpicture}
    \draw (0,0) circle (1cm);

    \node[draw,circle,fill=black,scale=.3] (x2) at (45:1) {};
    \node[draw,circle,fill=black,scale=.3] (x4) at (135:1) {};
    \node[draw,circle,fill=black,scale=.3] (x1) at (225:1) {};
    \node[draw,circle,fill=black,scale=.3] (x3) at (315:1) {};

    \node[draw,circle,fill=black,scale=.3,blue] (tau) at (270:1) {};
    \node[draw,circle,fill=black,scale=.3,blue] (sigma) at (90:1) {};

    \node[left] at (x1) {$x_1$};
    \node[right] at (x2) {$x_2$};
    \node[right] at (x3) {$x_3$};
    \node[left] at (x4) {$x_4$};
    \node[below,blue] at (tau) {$\tau$};
    \node[above,blue] at (sigma) {$\sigma$};

    \draw (x1) -- (x2);
    \draw (x3) -- (x4);
  \end{tikzpicture}
  \caption{\mbox{$x_1 < \tau < x_3 < x_2 < \sigma < x_4$}}
\end{subfigure}
\hspace{1cm}
  \begin{subfigure}[b]{0.4\textwidth}
    \centering
  \begin{tikzpicture}
    \draw (0,0) circle (1cm);

    \node[draw,circle,fill=black,scale=.3] (x2) at (45:1) {};
    \node[draw,circle,fill=black,scale=.3] (x4) at (135:1) {};
    \node[draw,circle,fill=black,scale=.3] (x1) at (225:1) {};
    \node[draw,circle,fill=black,scale=.3] (x3) at (315:1) {};

    \node[draw,circle,fill=black,scale=.3,blue] (tau) at (270:1) {};
    \node[draw,circle,fill=black,scale=.3,blue] (sigma) at (0:1) {};

    \node[left] at (x1) {$x_1$};
    \node[right] at (x2) {$x_2$};
    \node[right] at (x3) {$x_3$};
    \node[left] at (x4) {$x_4$};
    \node[below,blue] at (tau) {$\tau$};
    \node[right,blue] at (sigma) {$\sigma$};

    \draw (x1) -- (x2);
    \draw (x3) -- (x4);
  \end{tikzpicture}
  \caption{$x_1 < \tau < x_3 < \sigma < x_2 < x_4$}
\end{subfigure}
\hspace{1cm}
  \begin{subfigure}[b]{0.4\textwidth}
    \centering
  \begin{tikzpicture}
    \node (phantom) at (270:1) {};
    \node[below] at (phantom) {$\phantom{\tau}$};

    \draw (0,0) circle (1cm);

    \node[draw,circle,fill=black,scale=.3] (x2) at (45:1) {};
    \node[draw,circle,fill=black,scale=.3] (x4) at (135:1) {};
    \node[draw,circle,fill=black,scale=.3] (x1) at (225:1) {};
    \node[draw,circle,fill=black,scale=.3] (x3) at (315:1) {};

    \node[draw,circle,fill=black,scale=.3,blue] (tau) at (0:1) {};
    \node[draw,circle,fill=black,scale=.3,blue] (sigma) at (90:1) {};

    \node[left] at (x1) {$x_1$};
    \node[right] at (x2) {$x_2$};
    \node[right] at (x3) {$x_3$};
    \node[left] at (x4) {$x_4$};
    \node[right,blue] at (tau) {$\tau$};
    \node[above,blue] at (sigma) {$\sigma$};

    \draw (x1) -- (x2);
    \draw (x3) -- (x4);
  \end{tikzpicture}
  \caption{$x_1 < x_3 < \tau < x_2 < \sigma < x_4$}
\end{subfigure}
\hspace{1cm}
  \begin{subfigure}[b]{0.4\textwidth}
    \centering
  \begin{tikzpicture}
    \draw (0,0) circle (1cm);

    \node[draw,circle,fill=black,scale=.3] (x2) at (45:1) {};
    \node[draw,circle,fill=black,scale=.3] (x4) at (135:1) {};
    \node[draw,circle,fill=black,scale=.3] (x1) at (225:1) {};
    \node[draw,circle,fill=black,scale=.3] (x3) at (315:1) {};

    \node[draw,circle,fill=black,scale=.3,blue] (tau) at (-10:1) {};
    \node[draw,circle,fill=black,scale=.3,blue] (sigma) at (10:1) {};

    \node[left] at (x1) {$x_1$};
    \node[right] at (x2) {$x_2$};
    \node[right] at (x3) {$x_3$};
    \node[left] at (x4) {$x_4$};
    \node[right,blue] at (tau) {$\tau$};
    \node[right,blue] at (sigma) {$\sigma$};

    \draw (x1) -- (x2);
    \draw (x3) -- (x4);
  \end{tikzpicture}
  \caption{$x_1 < x_3 < \tau < \sigma < x_2 < x_4$}
\end{subfigure}
\par\bigskip
  \begin{subfigure}[b]{0.4\textwidth}
    \centering
  \begin{tikzpicture}
    \draw (0,0) circle (1cm);

    \node[draw,circle,fill=black,scale=.3] (x2) at (45:1) {};
    \node[draw,circle,fill=black,scale=.3] (x4) at (135:1) {};
    \node[draw,circle,fill=black,scale=.3] (x1) at (225:1) {};
    \node[draw,circle,fill=black,scale=.3] (x3) at (315:1) {};

    \node[draw,circle,fill=black,scale=.3,blue] (sigma) at (-10:1) {};
    \node[draw,circle,fill=black,scale=.3,blue] (tau) at (10:1) {};

    \node[left] at (x1) {$x_1$};
    \node[right] at (x2) {$x_2$};
    \node[right] at (x3) {$x_3$};
    \node[left] at (x4) {$x_4$};
    \node[right,blue] at (tau) {$\tau$};
    \node[right,blue] at (sigma) {$\sigma$};

    \draw (x1) -- (x2);
    \draw (x3) -- (x4);
  \end{tikzpicture}
  \caption{$x_1 < x_3 < \sigma < \tau < x_2 < x_4$}
\end{subfigure}
  \caption{All possible out-of-time-ordered configurations.}
  \label{fig:possible_OTO_configuration}
\end{figure}
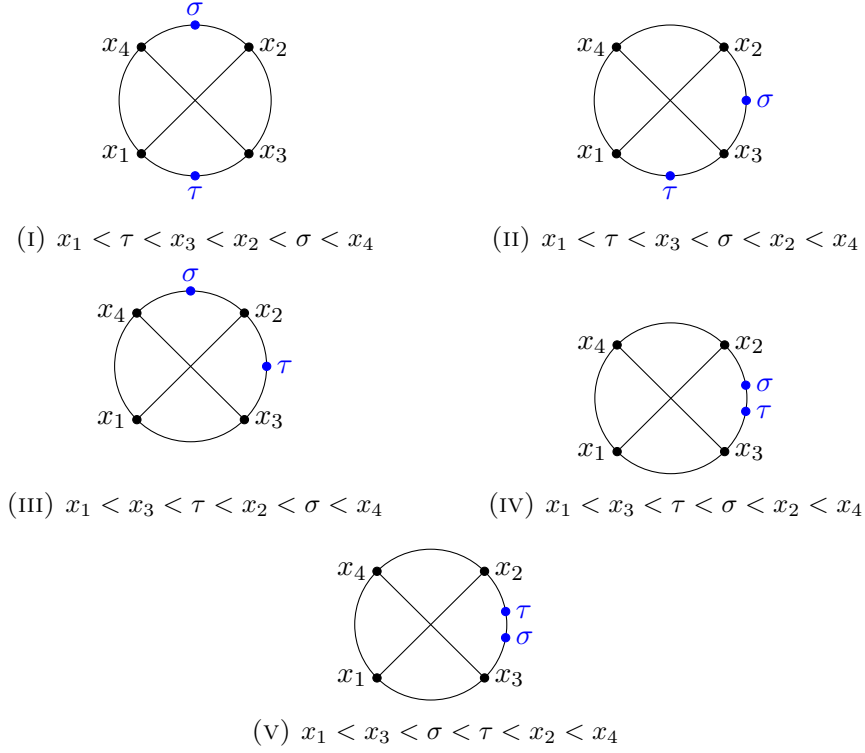
each of which is calculated as before by solving the equations of motion for the effective action \eqref{eq:mu_eff}.
The differences between these cases stems from the fact that the jumps of $u'_{cl}$ occur at different positions, and that eventually the expression $\exp(\mathcal S(u_{cl}, \sigma,\tau,t))$ has to be integrated over the correct domain. 
We give explicit results for each of the configurations  in Appendix \ref{app:OTO}.







\pagebreak
\appendix
\section{The Maurer-Cartan element of \texorpdfstring{$\vir$}{vir} and the KKS symplectic form on Virasoro orbits}\label{app:proofs_schwarzian}
Throughout this section we consider an orbit $\OO_{b}$ with representative \mbox{$b = b(x) dx^2$}.
We denote by 
\begin{equation*}
  b^f(x) = b(f(x)) f'^2 + \frac12 S(f)
\end{equation*}
 a generic point of $\OO_{b}$.

\begin{proposition}\label{prop:1}
  The one-form 
  \begin{equation}
    \hat Y = \frac{\delta f}{f'} + \delta k + \frac{c}{24}\int_0^{2\pi} dx \frac{\delta f}{f'} \left( \frac{f''}{f'} \right)'
  \end{equation}
  solves the equation 
  \begin{equation}
    \delta (b^f(x) - c) = ad^*(\hat Y) (b^f(x) - c).
  \end{equation}
\end{proposition}
\begin{proof} 
  Let $Y = \delta f / f'$.
  We recall that for any $L + k \in \vir$ the coadjoint action is given by
  \[
    ad^*(L + k) (b^f - c) = \Big( 2 L' b^f(x) + L (b^f(x))' + \frac{c}{12} L''' \Big) dx^2.
  \] 
  In particular the central terms act (and are acted up on) trivially and hence   
  \[
    ad^*(\hat Y) (b^f - c) =  ad^*(Y)(b^f).
  \] 
  From 
  \begin{equation}
    \begin{split}
    Y'&= \frac{\delta f'}{f'} - \frac{f''}{f'} \frac{\delta f}{f'} \\
    Y''&= \frac{\delta f''}{f'} - 2 \frac{f''}{f'} \frac{\delta f'}{f'} - \left( \left( \frac{f''}{f'} \right)' - \left( \frac{f''}{f'} \right)^2 \right) \frac{\delta f}{f'} 
  \end{split}
  \label{eq:Y}
  \end{equation}
  it follows that 
  \begin{align*}
    2Y' b^f((x)) &= \left( \frac{\delta f'}{f'} - \frac{f''}{f'}\frac{\delta f}{f'} \right)\left( 2b(f(x)) f'^2 + \frac{c}{6} S(f)\right) \\
    &= 2b(f(x))f'\delta f' - 2bf''\delta f + 2\cdot \frac{c}{12} S(f) Y'
  \end{align*}
  and
  \begin{align*}
    Y (b^f(x))' &= \frac{\delta f}{f'} \left(b(f(x)) f'^2 + \frac{c}{12}S(f) \right)' \\
    &= b'(f(x)) f'^2 \delta f + 2 b(f(x)) f'' \delta f + \frac{c}{12} S(f)' \frac{\delta f}{f'}.
  \end{align*}
  Therefore,
  \begin{align*}
    ad^*(Y) b^f &= 2Y' b^f(x) + Y (b^f(x))' + \frac{c}{12} Y'''  \\
    &=2b(f(x)) f' \delta f' + b'(f(x))f'^2 \delta f + \frac{c}{12} \left( 2 Y' S(f) + Y S(f)' + Y''' \right) \\
  \end{align*}
  The proof is completed by the following 
  \begin{lemma}\label{lem:1}
    $Y$ satisfies the following identity
    \begin{equation}
      \delta S(f) = 2S(f)Y' + YS(f)' + Y'''.
    \end{equation}
  \end{lemma}~\bigskip
\end{proof}

\begin{proof}[Proof of Lemma \ref{lem:1}]
  It will be convenient to rewrite the Schwarzian derivative of $f$ as follows:
  \begin{equation*}
    S(f) = \frac{f'''}{f'} - \frac32 \left( \frac{f''}{f'} \right)^2 = \left( \frac{f''}{f'} \right)' - \frac12 \left( \frac{f''}{f'} \right)^2.
  \end{equation*}
  For later use, we remark that 
  \begin{equation*}
    S(f)' = \left( \frac{f''}{f'} \right)'' - \left( \frac{f''}{f'} \right)' \left( \frac{f''}{f'} \right).
  \end{equation*}
  Then 
  \begin{equation}
    \begin{split}
      \delta S(f) &= \delta \left[ \left( \frac{f''}{f'} \right)' - \frac12 \left( \frac{f''}{f'} \right)^2 \right] \\
      &= \left( \frac{\delta f''}{f'} - \frac{f''\delta f'}{f'^2} \right)' - \left( \frac{f''}{f'} \right)\left( \frac{\delta f''}{f'} - \frac{f'' \delta f'}{f'^2} \right) \\
      &= \left( \frac{\delta f''}{f'} \right)' - \left( \frac{f''}{f'} \right)'\frac{\delta f'}{f'} - \frac{f''}{f'}\left( \frac{\delta f'}{f'} \right)' - \frac{f''}{f'}\frac{\delta f''}{f'} + \left( \frac{f''}{f'} \right)^2 \frac{\delta f'}{f'} \\
      &= \frac{\delta f'''}{f'} - 3 \frac{f''}{f'} \frac{\delta f''}{f'} - \left( S(f) - \frac32 \left( \frac{f''}{f'} \right)^2 \right)\frac{\delta f'}{f'}
    \end{split}
    \label{eq:delta_Sch}
  \end{equation}
  Now,  
  \begin{align*}
    Y'''&= \frac{\delta f'''}{f'} - 3 \frac{f''}{f'}\frac{\delta f''}{f'} -  3\left( \left( \frac{f''}{f'} \right)' - \left( \frac{f''}{f'}  \right)^2 \right)\frac{\delta f'}{f'} \\
    &+ \left( 3 \left( \frac{f''}{f'} \right)' \frac{f''}{f'} - \left( \frac{f''}{f'} \right)'' - \left( \frac{f''}{f'}  \right)^3\right) \frac{\delta f}{f'} \\
    &= \frac{\delta f'''}{f'} - 3 \frac{f''}{f'}\frac{\delta f''}{f'} -  3\left( S(f) - \frac12\left( \frac{f''}{f'}  \right)^2 \right)\frac{\delta f'}{f'} \\
    &+ \left(  2 \left( \frac{f''}{f'} \right)' \frac{f''}{f'} - S(f)' - \left( \frac{f''}{f'} \right)^3 \right) \frac{\delta f}{f'} \\
  \end{align*}
  which we can rewrite with the help of \eqref{eq:delta_Sch} and \eqref{eq:Y} as
  \begin{equation*}
    Y''' = \delta S(f) - 2S(f)Y' - S(f)'Y.
  \end{equation*}
\end{proof}

\begin{proposition}
  $\hat Y$ satisfies the Maurer-Cartan equation
  \begin{equation}
    \delta \hat Y - \frac12 [\hat Y,\hat Y] = 0.
  \end{equation}
\end{proposition}
\begin{proof}
  The proof goes by direct calculation. 
  To get the signs right, we recall that $\delta$ is interpreted as the exterior derivative on $\Diff^+(S^1)$ (resp.\ its central extension) and hence anti-commutes with the exterior derivative $d$ on $S^1$.
  Then one has 
  \begin{align*}
    \delta \hat Y &= \delta Y - \frac{1}{12}\int_0^{2\pi}dx\ \left( \delta Y \left( \frac{f''}{f'} \right)' - Y \delta \left( \frac{f''}{f'} \right)' \right) \\
    &= \delta Y - \frac{1}{24} \int_0^{2\pi}dx\ \left( \delta Y \left( \frac{f''}{f'} \right)' - Y \delta \left( S(f) + \frac12 \left( \frac{f''}{f'}  \right)^2\right) \right) \\
    &= \delta Y - \frac{1}{24} \int_0^{2\pi}dx\ \left( \delta Y \left( \frac{f''}{f'} \right)' - Y \delta S(f) - Y \left( \frac{f''}{f'} \right) \delta \left( \frac{f''}{f'}  \right). \right)
  \end{align*}
  Let us focus on the integral.
  By Lemma \ref{lem:1} and the fact that 
  \[
    \delta Y = \frac12[Y,Y]_{\vect(S^1)} = Y Y'
  \]
  one has 
  \begin{align*}
    &\int_0^{2\pi}dx\ \left( \delta Y \left( \frac{f''}{f'} \right)' - Y \delta S(f) - Y \left( \frac{f''}{f'} \right) \delta \left( \frac{f''}{f'}  \right) \right) =  \\
    =\;&\int_0^{2\pi}dx\ \left( \delta Y \left( \frac{f''}{f'} \right)' - Y \left( 2S(f)Y' + S(f')Y + Y''' \right) - Y \left( \frac{f''}{f'} \right) \delta \left( \frac{f''}{f'}  \right) \right)\\
    =\;&\int_0^{2\pi}dx\ \left( YY' \left( \left( \frac{f''}{f'} \right)' - 2S(f) \right) - YY'''- Y \left( \frac{f''}{f'} \right) \delta \left( \frac{f''}{f'}  \right) \right)\\
    =\;&\int_0^{2\pi}dx\ \left( YY' \left( \left( \frac{f''}{f'} \right)^2 - \left( \frac{f''}{f'} \right)' \right) - Y \left( \frac{f''}{f'} \right) \delta \left( \frac{f''}{f'}  \right) - YY'''\right).
  \end{align*}
  Notice that 
  \[
    YY' = \frac{\delta f \wedge \delta f'}{f'^2}.
  \] 
  A partial integration yields 
  \begin{align*}
    -\int_0^{2\pi}dx\ \frac{\delta f \wedge \delta f'}{f'^2} \left( \frac{f''}{f'} \right)' &=  \int_0^{2\pi}dx\ \left( \frac{\delta f \wedge \delta f'}{f'^2} \right)' \left( \frac{f''}{f'} \right) \\
    &= \int_0^{2\pi}dx\ \left( \frac{\delta f \wedge \delta f''}{f'^2} \left( \frac{f''}{f'} \right) - 2 \frac{\delta f \wedge \delta f'}{f'^2} \left( \frac{f''}{f'} \right)^2 \right)
  \end{align*}
  such that  
  \[
    \int_0^{2\pi}dx\ YY'\left( \left( \frac{f''}{f'}  \right)^2 - \left( \frac{f''}{f'} \right)'\right) = \int_0^{2\pi}dx\ \left(  \frac{\delta f \wedge \delta f''}{f'^2}\left( \frac{f''}{f'} \right)^2 - \frac{\delta f \wedge \delta f'}{f'^2} \left( \frac{f''}{f'} \right)^2 \right).
  \]
  Using 
  \[
    \delta \left( \frac{f''}{f'} \right) = \frac{\delta f''}{f'} - \frac{f''}{f'} \frac{\delta f'}{f'},
  \] 
  one then has 
  \[
    \int_0^{2\pi}dx\ \left( YY' \left( \left( \frac{f''}{f'} \right)^2 - \left( \frac{f''}{f'} \right)' \right) - Y \frac{f''}{f'} \delta \left( \frac{f''}{f'} \right) \right) = 0.
  \] 
  Putting all together, it follows that 
  \begin{equation*}
    \delta \hat Y = \delta Y - \frac{1}{24} \int_0^{2\pi}dx\ YY''' = \delta Y + \frac{1}{24} \int_{0}^{2\pi}dx\ Y'''Y = \frac12[\hat Y, \hat Y]
  \end{equation*}
  by definition of the bracket, {\em c.f.}\ \eqref{eq:adj}.

\end{proof}

\begin{proposition}
  The symplectic form $\omega$ on the Virasoro coadjoint orbit $\OO_{b}$ can be written as 
  \begin{equation}
    \begin{split}
      \omega &= \frac12 \int_{S^1}\left( \delta b^f \wedge \frac{\delta f}{f'} \right)dx \\
      &= \int_{S^1}\left( b(f(x)) \delta f' \wedge \delta f + \frac{c}{24} \delta \log(f') \wedge \delta \log(f')' \right)dx.
    \end{split}
  \end{equation}
\end{proposition}

\begin{proof}
  From \eqref{eq:delta_Sch} it then follows that 
  \begin{align*}
    &\int \left( \delta S(f) \wedge \frac{\delta f}{f'}  \right)dx \\
    =\;& \int \left[ \left( \frac{\delta f}{f'} \right)' \frac{\delta f''}{f'}  + \left( \frac{f''}{f'} \right)' \frac{\delta f \wedge \delta f'}{f'^2} - \left( \frac{f''}{f'} \frac{\delta f}{f'} \right)'\wedge \left( \frac{\delta f'}{f'} \right) \right.\\
    &\quad\!\quad\left. + \frac{f''}{f'} \frac{\delta f \wedge f'}{f'^2} - \left( \frac{f''}{f'} \right)^2\frac{\delta f \wedge \delta f'}{f'^2} \right]dx  + bdry\\
    =\;&  \int \left( \frac{\delta f' \wedge \delta f''}{f'^2} \right)dx + bdry \\
    =\;&  \int \left( \delta \log f' \wedge \delta (\log f')'  \right)dx + bdry
  \end{align*}
  where the first and third term in the second line arise from an integration by parts.
  The boundary term is given by 
  \[
    bdry = - \left.\frac{\delta f \wedge \delta f''}{f'^2} + \frac{f''}{f'} \frac{\delta f \wedge \delta f'}{f'^2} \right\rvert_{0}^{2\pi}.
  \]
  Recall that $f$ is quasi-periodic, $f(x+2\pi) = f(x) + 2\pi$. 
  Hence $f',f'', \delta f, \delta f'$ are $2\pi$-periodic and the boundary term vanishes. 
  From  
  \[
    \delta \left( b(f(x))f'^2\right) \wedge \frac{\delta f}{f'}  = (2b(f(x))f' \delta f' + b'(f(x))f'^2 \delta f) \wedge \frac{\delta f}{f'}= 2b(f(x)) \delta f' \wedge \delta f
  \]
  it then follows that 
  \begin{align*}
    \omega &= \frac{1}{2} \int \left( \delta b^f(x) \wedge \frac{\delta f}{f'} \right)dx \\
    &= \frac{1}{2} \int \left( \delta \left( b(f(x)) f'^2 + \frac{c}{12} S(f) \right) \wedge \frac{\delta f}{f'} \right) dx \\
    &= \int \left( b(f(x)) \delta f' \wedge \delta f + \frac{c}{24} \delta \log f' \wedge \delta (\log f')' \right)dx
  \end{align*}
\end{proof} 

\section{Explicit expression of out-of-time-ordered correlation functions}\label{app:OTO}
Below, we present explicit expressions for the out-of-time-ordered correlation function $\left\langle \OO(x_1,x_2)\OO(x_3,x_4) \right\rangle$, with $x_1 < x_3 < x_2 < x_4$.
We give the analytic result case by case, {\em c.f.}\ Figure \ref{fig:possible_OTO_configuration}.\bigskip

\paragraph{\textbf{Case (I)}}
This scenario is analogous to the time-ordered-correlator but with $x_2$ and $x_3$ interchanged, since neither $\tau$ nor $\sigma$ are positioned in the overlap of the intervals.
Hence,
\begin{equation*}
  \left\langle \OO(x_1,x_2) \OO(x_3,x_4) \right\rangle_{(I)} = \frac{\pi}{2} e^{-2\pi b_0 t + \frac{3}{2ct}(x_{42} + x_{31})} \sqrt{\frac{c t}{12}} \cdot \Big[  E(x_{31},x_{42}) + G(x_{31},x_{42}) \Big]
\end{equation*}
where we recall that
\begin{equation*}
  E(x,y) = 4 \sqrt{\frac{ct}{12}} \sum_{\varepsilon_1,\varepsilon_2 = \pm 1} \varepsilon_1 \varepsilon_2 \exp\left( -\frac{3}{4\pi ct} \left(\varepsilon_1 x + \varepsilon_2 y - \frac{c\pi t}{3} \alpha\right)^2  \right)
\end{equation*}
and
\begin{equation*}
  G(x,y) = \hspace{-1em}\sum_{\varepsilon_1, \varepsilon_2 = \pm 1}\hspace{-.5em} \varepsilon_1 \varepsilon_2 \erf\left( \sqrt{-2\pi b_0 t} - (\varepsilon_1 x + \varepsilon_2 y)\sqrt{\frac{3}{4\pi c t}} \right) \left( \frac{\pi ct}{3} \alpha - \varepsilon_1 x - \varepsilon_2 y \right)
\end{equation*}

\paragraph{\textbf{Case (II)}}
\begin{align*}
  &\left\langle \OO(x_1,x_2) \OO(x_3,x_4) \right\rangle_{(II)} = 2\pi e^{-2\pi b_o t} \left( \frac{ct}{12} \right)^{3/2} \Bigg[ \\
    & e^{ \frac{3}{2ct}\left( x_{43} + x_{21} - 2x_{23} \right)} \erf\left( \sqrt{-2\pi b_0 t} + (x_{43} - x_{21}) \sqrt{\frac{3}{4\pi c t}} \right)\\
    &-e^{ \frac{3}{2ct}\left( x_{43} + x_{21} - 2x_{23} \right)} \erf\left( \sqrt{-2\pi b_0 t} + (x_{43} + x_{21} - 2x_{23}) \sqrt{\frac{3}{4\pi c t}} \right) \\
    &-e^{ \frac{3}{2ct}\left(  x_{43} + x_{21} + 2 x_{23} \right)} \erf\left( \sqrt{-2\pi b_0 t} +  (x_{43} - x_{21}  + 2x_{23}) \sqrt{\frac{3}{4\pi c t}} \right) \\
    &+e^{ \frac{3}{2ct}\left( x_{43} + x_{21} + 2 x_{23} \right)} \erf\left( \sqrt{-2\pi b_0 t} + (x_{21} + x_{43}) \sqrt{\frac{3}{4\pi c t}} \right)\\
    &+e^{ -\frac{3}{2ct}\left( \frac{2\pi}{3} ct \alpha - 2\pi + x_{43} - 3 x_{21} + 2x_{23} \right)} \erf\left( \sqrt{-2\pi b_0 t} + \left( x_{43} - x_{21} + 2x_{23} - 2\pi \right)\sqrt{\frac{3}{4\pi ct}} \right)\\
    &-e^{ -\frac{3}{2ct}\left( \frac{2\pi}{3} ct \alpha - 2\pi + x_{43} - 3 x_{21} + 2x_{23} \right)} \erf\left( \sqrt{-2\pi b_0 t} + (x_{43} - x_{21} -2\pi) \sqrt{\frac{3}{4\pi c t}} \right) \\
    &-e^{ -\frac{3}{2ct}\left( \frac{2\pi}{3} ct \alpha - 2\pi + x_{43} + x_{21} - 2 x_{23} \right)} \erf\left( \sqrt{-2\pi b_0 t} + (x_{43} + x_{21} - 2\pi) \sqrt{\frac{3}{4\pi ct}} \right) \\
  &+e^{ -\frac{3}{2ct}\left(  \frac{2\pi}{3} ct \alpha - 2\pi + x_{43} + x_{21} - 2 x_{23} \right)} \erf\left( \sqrt{-2\pi b_0 t} + (x_{43} + x_{21} - 2x_{23} - 2\pi)\sqrt{\frac{3}{4\pi c t}} \right) \Bigg]
\end{align*}
\vfill

\paragraph{\textbf{Case (III)}} 
\begin{align*}
  &\left\langle \OO(x_1,x_2) \OO(x_3,x_4) \right\rangle_{(III)} = 2\pi e^{-2\pi b_o t} \left( \frac{ct}{12} \right)^{3/2} \Bigg[ \\
    & e^{ \frac{3}{2ct}\left( x_{43} + x_{21} - 2x_{23} \right)} \erf\left( \sqrt{-2\pi b_0 t} - (x_{43} + x_{21} - 2 x_{23}) \sqrt{\frac{3}{4\pi ct}} \right) \\
    &-e^{ \frac{3}{2ct}\left( x_{43} + x_{21} - 2x_{23} \right)} \erf\left( \sqrt{-2\pi b_0 t} + (x_{43} - x_{21})\sqrt{\frac{3}{4\pi c t}} \right) \\
    &-e^{ \frac{3}{2ct}\left( x_{43} + x_{21} + 2x_{23} \right)} \erf\left( \sqrt{-2\pi b_0 t} - (x_{43} + x_{21})\sqrt{\frac{3}{4\pi c t}} \right) \\
    &+e^{  \frac{3}{2ct}\left( x_{43} + x_{21} + 2x_{23} \right)} \erf\left( \sqrt{-2\pi b_0 t} + (x_{43} - x_{21} - 2x_{23})\sqrt{\frac{3}{4\pi c t}} \right) \\
    &-e^{ \frac{3}{2ct}\left( \frac{2\pi}{3} ct \alpha + 2\pi - x_{43} - x_{21} + x_{23} \right)} \erf\left( \sqrt{-2\pi b_0 t} - (x_{43} + x_{21} - 2 x_{23} - 2\pi)\sqrt{\frac{3}{4\pi c t}} \right) \\
    &+e^{ \frac{3}{2ct}\left( \frac{2\pi}{3} ct \alpha + 2\pi - x_{43} - x_{21} + x_{23} \right)} \erf\left( \sqrt{-2\pi b_0 t} - ( x_{43} + x_{21} - 2\pi)\sqrt{\frac{3}{4\pi c t}} \right) \\
    &+e^{ \frac{3}{2ct}\left( \frac{2\pi}{3} ct \alpha + 2\pi + 3 x_{43} - x_{21} - 2 x_{23} \right)} \erf\left( \sqrt{-2\pi b_0 t} + ( x_{43} - x_{21} + 2\pi)\sqrt{\frac{3}{4\pi c t}} \right) \\
  &-e^{ \frac{3}{2ct}\left( \frac{2\pi}{3} ct \alpha + 2\pi + 3 x_{43} - x_{21} - 2 x_{23} \right)} \erf\left( \sqrt{-2\pi b_0 t} + ( x_{43} - x_{21} - 2 x_{23} + 2\pi)\sqrt{\frac{3}{4\pi c t}} \right)\Bigg]
\end{align*}
\vfill

\paragraph{\textbf{Case (IV)}} 
\begin{align*}
  &\left\langle \OO(x_1,x_2) \OO(x_3,x_4) \right\rangle_{(IV)} = \pi e^{-2\pi b_o t} \left( \frac{ct}{12} \right)^{3/2} \Bigg[ \\
    & e^{ \frac{3}{2ct}\left( x_{43} + x_{21} + 2 x_{23} \right)} \erf\left( \sqrt{-2\pi b_0 t} + (x_{43} - x_{21} + 2 x_{23}) \sqrt{\frac{3}{4\pi ct}} \right) \\
    &-e^{ \frac{3}{2ct}\left( x_{43} + x_{21} + 2 x_{23} \right)} \erf\left( \sqrt{-2\pi b_0 t} + (x_{43} - x_{21} - 2 x_{23})\sqrt{\frac{3}{4\pi c t}} \right) \\
    &-e^{ -\frac{3}{2ct}\left( \frac{2\pi}{3} c t \alpha - 2\pi + x_{43} - 3x_{21} - 2x_{23} \right)} \erf\left( \sqrt{-2\pi b_0 t} + (x_{43} - x_{21} + 2 x_{23} - 2\pi)\sqrt{\frac{3}{4\pi c t}} \right) \\
    &+e^{ -\frac{3}{2ct}\left( \frac{2\pi}{3} c t \alpha - 2\pi + x_{43} - 3x_{21} - 2x_{23} \right)} \erf\left( \sqrt{-2\pi b_0 t} + (x_{43} - x_{21} - 2\pi)\sqrt{\frac{3}{4\pi c t}} \right) \\
    &-e^{ \frac{3}{2ct}\left( \frac{2\pi}{3} c t \alpha + 2\pi + 3 x_{43} - x_{21} - 2 x_{23} \right)} \erf\left( \sqrt{-2\pi b_0 t} + (x_{43} - x_{21} + 2\pi)\sqrt{\frac{3}{4\pi c t}} \right) \\
  &+e^{  \frac{3}{2ct}\left( \frac{2\pi}{3} c t \alpha + 2\pi + 3 x_{43} - x_{21} - 2 x_{23} \right)} \erf\left( \sqrt{-2\pi b_0 t} + (x_{43} - x_{21} - 2x_{23} + 2\pi)\sqrt{\frac{3}{4\pi c t}} \right)\Bigg]
\end{align*}

\paragraph{\textbf{Case (V)}} 
\begin{align*}
  &\left\langle \OO(x_1,x_2) \OO(x_3,x_4) \right\rangle_{(V)} = 2\pi e^{-2\pi b_o t} \left( \frac{ct}{12} \right)^{3/2} \Bigg[ \\
    & e^{ \frac{3}{2ct}\left( x_{43} + x_{21} + 2 x_{23} \right)} \erf\left( \sqrt{-2\pi b_0 t} + (x_{43} - x_{21} + 2x_{23}) \sqrt{\frac{3}{4\pi ct}} \right) \\
    &-e^{ \frac{3}{2ct}\left( x_{43} + x_{21} + 2 x_{23} \right)} \erf\left( \sqrt{-2\pi b_0 t} + (x_{43} - x_{21} - 2x_{23}))\sqrt{\frac{3}{4\pi c t}} \right) \\
    &-e^{ -\frac{4}{2ct}\left( \frac{2\pi}{3}ct\alpha - 2\pi + x_{43} - 3x_{21} + 2x_{23} \right)} \erf\left( \sqrt{-2\pi b_0 t} + (x_{43} - x_{21} + 2x_{23} - 2\pi)\sqrt{\frac{3}{4\pi c t}} \right) \\
    &+e^{ -\frac{4}{2ct}\left( \frac{2\pi}{3}ct\alpha - 2\pi + x_{43} - 3x_{21} + 2x_{23} \right)} \erf\left( \sqrt{-2\pi b_0 t} + (x_{43} - x_{21} - 2\pi)\sqrt{\frac{3}{4\pi c t}} \right) \\
    &-e^{ \frac{3}{2ct}\left( \frac{2\pi}{3} ct\alpha + 2\pi + 3 x_{43} - x_{21} - x_{23} \right)} \erf\left( \sqrt{-2\pi b_0 t} + ( x_{43} - x_{21} + 2\pi)\sqrt{\frac{3}{4\pi c t}} \right) \\
  &+e^{ \frac{3}{2ct}\left( \frac{2\pi}{3} ct\alpha + 2\pi + 3 x_{43} - x_{21} - x_{23} \right)} \erf\left( \sqrt{-2\pi b_0 t} + ( x_{43} - x_{21} - 2x_{23} + 2\pi)\sqrt{\frac{3}{4\pi c t}} \right)\Bigg]
\end{align*}

\end{document}